\documentclass[journal,doublecolumn]{IEEEtran}
\usepackage{cite}
\usepackage[english]{babel}

\usepackage{tikz}
\usepackage{float}
\usepackage{amsthm}
\usepackage{amssymb}
\usepackage{commath}
\usepackage{mathrsfs}
\usepackage{mdframed}
\usepackage{multirow}
\usepackage{graphicx}
\usepackage{booktabs}
\usepackage{tabularx}
\usepackage{subcaption}
\usepackage[colorlinks=true, allcolors=blue]{hyperref}
\usepackage{array}
\usepackage{comment}
\usepackage{relsize,amsfonts}
\usepackage{mathtools}
\usepackage[ruled,lined,linesnumbered]{algorithm2e}
\usetikzlibrary{decorations.pathreplacing,calc}

\newcolumntype{C}{>{\centering\arraybackslash}X}

\DeclareSymbolFont{myletters}{OML}{ztmcm}{m}{it}
\DeclareMathSymbol{\uplambda}{\mathord}{myletters}{"15}
\usepackage{filecontents,lipsum}

\theoremstyle{definition}

\newtheorem{proposition}{Proposition}
\newtheorem{lemma}{Lemma}
\newtheorem{remark}{Remark}
\newtheorem{theorem}{Theorem}
\newtheorem{corollary}{Corollary}
\newtheorem{assumption}{Assumption}

\newcommand{\E}{\mathbb{E}}

\newcommand{\cQ}{\mathcal{Q}}

\newcommand{\cE}{\mathcal{E}}
\newcommand{\cF}{\mathcal{F}}

\newcommand{\cS}{\mathcal{S}}
\newcommand{\vb}{\mathbf{b}}
\newcommand{\vw}{\mathbf{w}}
\newcommand{\vg}{\mathbf{g}}
\newcommand{\ve}{\mathbf{e}}
\newcommand{\vv}{\mathbf{v}}

\title{Pinching-antenna-enabled Federated Learning: Tail Latency, Participation, and Convergence Analysis}

\author{Yushen Lin,~\IEEEmembership{Student Member,~IEEE,}
{Zihan Chen,}~\IEEEmembership{Member,~IEEE,}
       and {Zhiguo Ding,~\IEEEmembership{Fellow,~IEEE}}
\vspace{-2em} 
\thanks{Y. Lin is with the School of Electrical and Electronic Engineering, The University of Manchester, M13 9PL, U.K. (e-mail: yushen.lin@manchester.ac.uk).

Zihan Chen is with the Information Systems Technology and Design Pillar, Singapore University of Technology and Design, Singapore. (e-mail: zihan\_chen@mymail.sutd.edu.sg)

Zhiguo Ding is with the Department of Electrical and Electronic Engineering,
University of Manchester, Manchester, UK, and the Department of Computer and Information Engineering, Khalifa University, Abu Dhabi, UAE (e-mail: zhiguo.ding@manchester.ac.uk).}
}

\begin{document}

\maketitle

\begin{abstract}
Federated learning (FL) in wireless networks is limited by straggler delays from unpredictable channel conditions. 
In this paper, we investigate the pinching‑antenna system (PASS), which dynamically “pinches” the radiator along a dielectric waveguide to shorten the worst links. 
In synchronous FL (SFL), we prove that PASS shortens the worst‑link distance, and it increases the on‑time completion probability in asynchronous FL (AFL). 
Accordingly, SFL exhibits stochastic dominance on round time, while AFL yields explicit latency and participation gains. We then pair physical‑layer (PHY)-aware sampling with error-feedback compression and prove that pinching raises the minimum inclusion probability, thus shrinking both the sampling variability and compression-induced floors in a Lyapunov analysis.
Simulations demonstrate consistent wall clock speedups and markedly shorter latency tails. By addressing stragglers at their PHY root, PASS complements higher‑layer scheduling and accelerates wireless FL in both SFL and AFL.
\end{abstract}

\vspace{-0.5em}
\section{Introduction}
Wireless federated learning (FL) promises training on the device without sharing raw data, but the unreliability of the radio links creates stragglers that dominate the wall clock time \cite{FL_Wireless_Qin}.
In synchronous FL (SFL), per‑round latency is set by the slowest uplink \cite{Straggler_FL}. In asynchronous federated learning (AFL), timely completion within a deadline governs participation \cite{Straggler_literature}. 
A wide range of higher-layer approaches—e.g., asynchronous protocols, intelligent user selection, and resource scheduling—has been proposed to mitigate straggler symptoms at the system level \cite{ASGD, Lin_TWC, FL_liteature}. 
We pose a complementary question: \emph{Can we address stragglers at physical-layer (PHY) root?}

This question has gained practical relevance with the advent of the pinching-antenna system (PASS) \cite{PinchingAntenna_DOCOMO, Ding_PASS_Original}. As a reconfigurable antenna technology, PASS provides an unprecedented ability to manipulate the propagation environment by creating targeted line-of-sight (LoS) links on demand. In effect, PASS can transform large-scale channel gain from a random quantity into a tunable parameter \cite{Ouyang_array_gain}. 
By markedly improving the worst wireless links, PASS has the potential to preempt stragglers before they dominate training latency. 

The goal of this paper is to go beyond intuition and establish a principled foundation for quantifying the impact of PASS on both latency and convergence. 
SFL is critical-path limited by the slowest uplink, and AFL is deadline or throughput limited with completion times. Upper-layer approaches, i.e., asynchrony, user selection, and resource scheduling, treat the symptom of stragglers by altering participation sets or deadlines, whereas PASS treats the physical cause by shortening worst links. 
Since wireless stragglers arise from path loss and fading, a PA-equipped system can shrink worst-link distances in SFL and raise timely completions in AFL. Studying both regimes clarifies where PASS helps from uniform stochastic dominance on round time (SFL) to explicit latency gains and participation lifts (AFL), and how these gains map to convergence.
Our analysis shows that PASS increases and thus reduces the sampling factor regardless of the scheduler, so any scheduler that benefits from higher timely participation inherits these gains.
The main contributions are summarized as follows:
\begin{enumerate}
\item For SFL, we prove that the bottleneck distance (the farthest selected user) under PASS is never larger than conventional BS, and is strictly smaller when user locations are continuously distributed, establishing stochastic dominance on round time.
We derive explicit non-asymptotic latency gains with a computable SNR threshold, showing PASS achieves strictly positive time savings in the high-SNR regime for AFL.
\item  We establish that PASS yields no fewer participating users than conventional systems for any spatial distribution under a given per-round latency budget, with closed-form participation gains demonstrating larger benefits in highly clustered or widely dispersed scenarios where conventional systems suffer severe straggler effects.
\item We show that PASS raises the minimum inclusion probability $\pi_{\text{min}}$ by enlarging the eligible-user set each round, which directly reduces the Horvitz–Thompson (HT) amplification factor and lowers both sampling variance and compression-induced error floors in the aggregation mechanism.
\item Through a Lyapunov analysis with explicit stepsize and weighting conditions, we separate gradient-variance floors from compression floors and prove that, under bounded staleness, PASS enables higher update rates and larger stable stepsizes, yielding faster wall-clock convergence in both SFL and AFL regimes.
\end{enumerate}

\vspace{-0.5em}
\section{Related Work}
\subsection{Recent Developments in PASS}
Recent works have advanced the design and understanding of PASS. For instance, Zeng et al. \cite{PA_zeng_literature} highlighted that optimizing pinching antenna placement together with power, time slot, and subcarrier allocation is critical to fully exploit PASS channel reconfigurability. 
Ding et al. \cite{PA_ding_place} presented closed-form solutions to the PA-placement problem, showing that the fairness-optimal OMA placement is central, while in NOMA the optimal position skews toward the nearest user to enhance performance.
Wang et al. \cite{PA_wang_literature} developed a physics-based model for PASS and formulated a joint transmit and pinching-beamforming optimization, demonstrating that PASS can reduce the required transmit power by over 95\% compared to a conventional massive MIMO baseline for the same performance. Ouyang et al. \cite{Ouyang_array_gain} derive closed-form expressions for the array gain of multi-antenna PASS setups, showing that by optimizing the number and spacing of pinching antennas along the waveguide, PASS can achieve much higher array gains than traditional antenna arrays. 
Moreover, in a downlink multi-user design, Wang et al. \cite{kaidi_PA} showed that a NOMA-assisted PASS with multiple dynamically activated PAs yields notably higher sum rates than a conventional fixed-antenna system. For the uplink, Tegos et al. \cite{PA_literature_Tegos} investigated the uplink of PASS and proposed jointly optimizing PA positions and radio resources to maximize the minimum user data rate, achieving robust fairness gains and outperforming conventional systems in worst-case user throughput.
These advances underscore the potential of PASS to transform wireless networks by enabling on-demand, near-LoS links, and they motivate exploring PASS in diverse communication paradigms beyond traditional setups.

\subsection{FL in Wireless Networks}
Stragglers have been addressed by asynchrony with bounded staleness and sharp analyses, user selection/scheduling to trade off latency vs. bias, and model adaptation, e.g., partial model training for weak users \cite{bounded_staleness_literature, client_selection_literature, partial_model_training}. On the other hand, a major issue of wireless FL is the communication bottleneck when many users upload large model updates over a shared channel. To alleviate this, researchers have developed over-the-air (OTA) aggregation techniques that exploit the signal-superposition property of the wireless medium. For example, 
Zhu et al. \cite{OTA_literature} proposed a one-bit OTA-FL scheme where users quantize gradients to 1-bit and transmit over a common channel with digital modulation, using a majority-vote decoding at the server. Another line of work focuses on optimizing resource usage for FL in wireless environments. 
Beyond communication efficiency, the straggler problem, where slow users delay FL rounds, has prompted innovative solutions. 
Liu et al. \cite{FL_literature_2_adaptive} introduced an adaptive clustering strategy that jointly optimizes computation and communication to minimize total user energy consumption under a strict FL latency deadline.
Wu et al. \cite{FL_literature_3_straggler} tackled heterogeneity by allowing weaker users to train only a partial model instead of the full model, which enables all users to contribute according to their capacity, reaching the target accuracy faster than standard FedAvg.
These developments from communication-efficient aggregation to straggler-aware scheduling have greatly improved the scalability and resilience of wireless FL.

Motivated by the above advances, integrating PASS with FL emerges as a promising approach to address the straggler problem at its physical-layer root while leveraging higher-layer optimizations \cite{wu_PA_FL_2025}. By proactively strengthening the worst wireless links, a PASS-equipped network can prevent extreme delays before they occur. Recent work by Wu et al. \cite{wu_PA_FL_2025}, demonstrated that dynamically deploying PAs alongside conventional antennas can effectively mitigate wireless stragglers by establishing strong LoS links on demand. In their hybrid PASS-conventional design, a fuzzy-logic user grouping and a deep reinforcement learning scheduler jointly optimize PA placement and resource allocation, yielding faster FL convergence and lower round latency than baseline scheduling. Aside from this initial study, the convergence of PASS with FL remains largely unexplored, highlighting the importance of investigating PA-assisted FL in future wireless networks.

\noindent\textbf{\textit{Notation}.} Let $\mathcal{U}[a,b]$ denote the uniform distribution on $[a,b]$. Let $\mathrm{Beta}(\cdot,\cdot)$ denote the Beta distribution. The statistical expectation operator is denoted by $\E[\cdot]$. Denote $\Phi(\cdot)$ as the standard Gaussian cumulative distribution function (CDF). For coordinates $\{X_i\}_{i=1}^K$, write $X_{(i)}$ for the $i$th order statistic (ascending). Define $Y_i \triangleq |X_i|$ and write $Y_{[i]}$ for the $i$th order statistic of $\{Y_i\}_{i=1}^K$.
Let $d_w$ denote the model dimension.
Let $K$ denote the total number of users, and $T_d$ as global per-upload deadline. 
In SFL, the server schedules $M$ users per round with $1\le M\le K$, and we write $\cS_t$ for the scheduled set at round $t$ with $|\cS_t|=M$.
In AFL, uploads are single-user per transmission; we keep the same notation for consistency.

\vspace{-0.5em} 
\section{Using Pinching Antenna in FL}
\subsection{System Model}
Consider an uplink scenario where users move along the $x$-axis, which models cases such as movement along a one-way traffic lane or a railway platform. The horizontal coordinate of the user $i$ follows $X_i\sim\mathcal{U}\left[-\frac{D}{2},\frac{D}{2}\right]$. Let $x\in[-\frac{D}{2},\frac{D}{2}]$ denote the horizontal coordinate of the user and $z\in[-\frac{D}{2},\frac{D}{2}]$ the coordinate of the active radiator along the waveguide. 
Consider an uplink from a single-antenna user at horizontal coordinate $x$ to an active radiator located at $z$ along a dielectric waveguide mounted at height $d>0$. The link distance is
$r(x,z)\triangleq\sqrt{(x-z)^2+d^2}$. For reference purposes, $R_0(x)\equiv R(x,0)$ represents the spectral efficiency achieved with a conventional fixed-antenna configuration. For PA, $z$ is decision‑dependent as specified later.
Under a spherical-wave LoS model, the complex baseband channel is \cite{Ding_PASS_Original,Ouyang_array_gain}
\begin{equation}\label{eq:channel}
  h(x,z)=\frac{\sqrt{\eta_f}\,e^{-j k_0 r(x,z)}}{r(x,z)},
\end{equation}
where $\eta_f\triangleq \frac{c_0^2}{16\pi^2 f_c^2}$, $c_0$ is the speed of light, $f_c$ denotes the carrier frequency, and $k_0 = \frac{2\pi}{\lambda_w}$ is the wavenumber with $\lambda_w=\frac{c_0}{f_c}$.
With transmit power $P$ and noise power $\sigma_n^2$, and denoting by $\phi$ the in-waveguide excitation/phasor, the instantaneous SNR is given by
\(
  \frac{P}{\sigma_n^2}\,\bigl\|h(x,z)\bigr\|^2 .
\)
The achievable spectral efficiency is
\begin{equation}\label{eq:Rfull}
  R(x,z)=\log_2\!\Bigl(1+\frac{P}{\sigma_n^2}\,\bigl\|h(x,z)\bigr\|^2\Bigr).
\end{equation}
Throughout, each uplink activates a single radiator. In AFL, the PA pins to the scheduled user so that $x-z = 0$. In SFL, the PA selects one radiator position $z^*$ per round and keeps it fixed; $z^*$ is chosen as the midpoint of the tightest window covering the $M$ scheduled users.
Let $\mathcal{S} \coloneqq \frac{P\,\eta_f}{\sigma_n^2}$. Then \eqref{eq:Rfull} becomes $ R(x,z) = \log_2\left(1 + \frac{\mathcal{S}}{(x-z)^2 + d^2}\right)$.

The per‑upload latency can be expressed as
\(
  \tau(x,z)=\frac{c}{R(x,z)},
\)
with $c = \frac{B_t}{\Delta W}$, where $B_t=d_w b_t$ is the payload size for a $d_w$ dimensional model with $b_t$ bits per coordinate, $W$ is the uplink bandwidth and $\Delta$ denotes the scaling parameter, i.e., in SFL with FDMA, $\Delta = \frac{1}{M}$ (equal bandwidth split $W/M$); in AFL,  $\Delta = 1$.

\begin{figure}[t]
\centering
\begin{tikzpicture}[
    >=stealth,
    scale=0.90,
    every node/.style={font=\footnotesize},
    safeTag/.style={fill=white,rounded corners=1pt,inner sep=1.2pt,outer sep=0pt},
    safeLbl/.style={fill=white,inner sep=1pt,outer sep=0pt},
    radiatorPA/.style={fill=black,draw=black,line width=0.6pt},
    wave/.style={line width=1.0pt},
    linkPA/.style={line width=0.9pt},
]
\def\panelgap{2.30cm}

\tikzset{
  usericon/.pic={
    \fill (0,0.10) circle (0.055); 
    \draw[line width=0.5pt] (-0.14,-0.02) .. controls (-0.10,-0.16) and (0.10,-0.16) .. (0.14,-0.02); 
  }
}

\newcommand{\BS}[2]{%
  \begin{scope}[shift={(#1,#2)}]
    \draw[rounded corners=1pt,fill=white] (-0.38,0) rectangle (0.38,0.34);
    \node at (0,0.17) {\scriptsize BS};
    \draw (0,0.34)--(0,0.62);
    \fill (0,0.62) circle (0.02);
    \draw (0.00,0.62) ++(0,0) arc[start angle=90,end angle=35,radius=0.22];
    \draw (0.00,0.62) ++(0,0) arc[start angle=90,end angle=20,radius=0.32];
    \draw[wave] (0.38,0.17)--(0.62,0.17);
  \end{scope}
}

\begin{scope}[yshift=\panelgap]
  \node[anchor=west, safeTag] at (-3.0,1.12) {(a) SFL};

  \BS{-3.80}{0.95}
  \coordinate (feedA) at (-3.18,1.12);
  \draw[wave] (feedA)--(-3.02,0.80);
  \draw[wave] (-3.02,0.80)--(3.00,0.80);
  \node[anchor=east, safeLbl] at (2.95,1.02) {waveguide};

  \draw (-3.00,-0.35)--(3.00,-0.35);
  \foreach \xx in {-2.25,-1.50,-0.80,-0.20,0.30,0.90,1.50,2.15}{
    \fill (\xx,-0.35) circle (0.05);
  }

  \def\winL{-0.80}
  \def\winR{ 0.90}

  \draw[densely dashed,line width=0.5pt] (0.00,-0.35)--(0.00,0.66);
  \draw[radiatorPA] (0.00,0.80) circle (0.08);
  \node[safeLbl,anchor=west] (pacenter) at (0.55,0.58) {PA centered};
  \draw[->] (pacenter.west)++(-0.06,0) -- (0.00,0.74);

  \draw[decorate,decoration={brace,amplitude=3pt}] (\winL,-0.06) -- (\winR,-0.06)
      node[midway,yshift=8pt,safeLbl]{tightest window};
  \draw[line width=0.5pt] (\winL,-0.31)--(\winL,-0.06);
  \draw[line width=0.5pt] (\winR,-0.31)--(\winR,-0.06);

  \coordinate (xu) at (\winR,-0.35);
  \draw[linkPA] (xu)--(0.00,0.80);
  \path (xu) ++(0,0.12) pic {usericon}; 
\end{scope}

\begin{scope}
  \node[anchor=west, safeTag] at (-3.0,1.12) {(b) AFL};

  \BS{-3.80}{0.95}
  \coordinate (feedB) at (-3.18,1.12);
  \draw[wave] (feedB)--(-3.02,0.80);
  \draw[wave] (-3.02,0.80)--(3.00,0.80);
  \node[anchor=east, safeLbl] at (2.95,1.02) {waveguide};

  \draw (-3.00,-0.35)--(3.00,-0.35);
  \foreach \xx in {-2.40,-1.40,-0.40,1.00,2.10}{
    \fill (\xx,-0.35) circle (0.05);
  }

  \coordinate (xu2) at (1.00,-0.35);
  \path (xu2) ++(0,0.12) pic {usericon};

  \draw[radiatorPA]   (1.00,0.80) circle (0.08);
  \draw[linkPA]       (xu2)--(1.00,0.80);

  \node[safeLbl,anchor=west] (paTxt) at (2.05,0.46) {PA over user};
  \draw[->] (paTxt.west)++(-0.08,0) -- (1.00,0.80);

  \node[safeLbl,anchor=west] at (1.35,-0.02) {meets $T_d$ \checkmark};
\end{scope}
\end{tikzpicture}

\caption{(a) SFL—PA centered on the tightest window; (b) AFL—PA over the user.}
\label{fig:system-two-panel}
\end{figure}
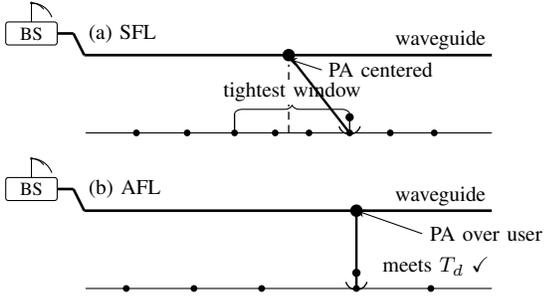

\vspace{-0.5em}
\subsection{FL Setup}
The federated optimization problem is formulated as the minimization of a global objective \cite{joint_communications_FL}:
\begin{equation}
    F(\vw)=\frac{1}{K}\sum_{i=1}^K f_i(\vw),
\end{equation}
where each local objective $f_i$ is $L$-smooth. 
The server schedules a set $\cS_t$ with $|\cS_t|=M$ and aggregates
\begin{equation}
    \widehat{\vg_t}=\frac{1}{M}\sum_{i\in\cS_t}\vg_{i,t},\qquad
    \vw_{t+1}=\vw_t-\eta_t\,\widehat{\vg_t}.
\end{equation}
The wall–clock round time is dominated by the slowest uplink in $\cS_t$.
In the communication round (or event) $t$, the user $i$ computes a stochastic gradient $\vg_{i,t}$ in the current global model $\vw_t$, satisfying:
\begin{align}
\E[\vg_{i,t}\mid\cF_t]=\nabla f_i(\vw_t),\quad
\E\!\big[\|\vg_{i,t}-\nabla f_i(\vw_t)\|^2\mid\cF_t\big]\le\sigma^2,
\end{align}
where $\cF_t$ represents the natural filtration up to round $t$, and $\sigma^2$ bounds the local gradient variance. 

The data heterogeneity across users is quantified by the condition \cite{Lin_TWC}:
\begin{align}
\frac{1}{K}\sum_{i=1}^K\|\nabla f_i(\vw)-\nabla F(\vw)\|^2\le \delta^2,
\end{align}
for all $\vw \in \mathbb{R}^{d_w}$, where $\delta^2$ represents the heterogeneity parameter. A larger $\delta^2$ indicates a higher degree of data heterogeneity.

The system operates in single-user uplink mode for each transmission in AFL. For example, a user participates if it completes within a deadline $T_d$. More details will be discussed in Section \ref{subsec:HT} later.
Under conventional ($\mathrm{CONV}$) baseline, the radiator remains fixed, with user positions uniformly distributed as $x\sim\mathcal U[-\frac{D}{2},\frac{D}{2}]$. under PA, the radiator pins to the user so $x-z=0$ and the link distance is $r = d$. In SFL, the setups for $\mathrm{CONV}$ and PA are similar, with differences discussed below.

\subsection{Conventional Antenna: Straggler Distance Analysis}
In SFL with $M$ scheduled users, the round latency is determined by the slowest participant. The optimal communication strategy under a fixed antenna configuration is achieved by selecting the $M$ users with the smallest distance offsets $|x|$, resulting in a straggler offset equal to the $M$‑th order statistic of ${|X_i|}$. This approach minimizes the maximum distance among any $M$ users, with the resulting round latency governed by the $M$-th smallest absolute offset.

Consider $K$ single‑antenna users on the $x$‑axis with i.i.d.\ positions
$X_1,\ldots,X_K \sim \mathcal{U}[-\frac{D}{2},\frac{D}{2}]$. Define $Y_i:=|X_i| \sim \mathcal{U}[0,\frac{D}{2}]$ and write $Y_{[1]}\le\cdots\le Y_{[K]}$ for their order statistics.

\paragraph*{Bottleneck distance under $\mathrm{CONV}$}
Selecting the $M$ smallest $|X_i|$ yields the per-round straggler offset as the $M$-th order statistic:
\begin{equation}
x_{\rm str}^{\rm CONV}=Y_{[M]}.
\label{eq:conv-xstr}
\end{equation}
\begin{remark}
The exact second moment of the $\mathrm{CONV}$ straggler distance:
\begin{equation}
\mathbb E\!\big[Y_{[M]}^{2}\big] 
= \Big(\tfrac{D}{2}\Big)^{2}\,\mathbb E\!\big[\widetilde Y_{[M]}^{2}\big]
= \Big(\tfrac{D}{2}\Big)^{2}\,\frac{M(M{+}1)}{(K{+}1)(K{+}2)}.
\label{eq:conv-YM-second-moment}
\end{equation}
\end{remark}
\begin{IEEEproof}
    It can be derived by using the mean value and second moments of $Y_{[M]}$ provided in Appendix \ref{appendix:Ym}.
\end{IEEEproof}

\begin{figure}
    \centering
    \includegraphics[width=0.8\linewidth]{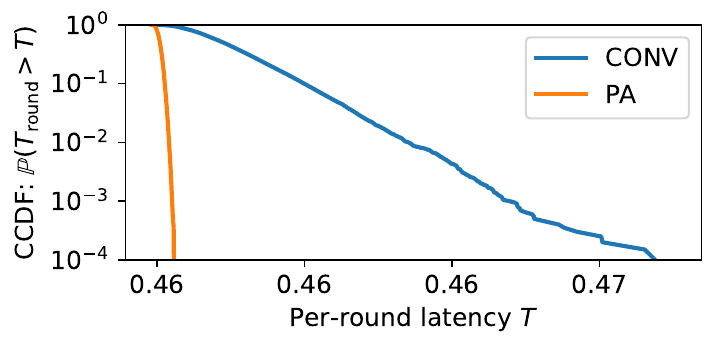}
    \caption{PA stochastically dominates $\mathrm{CONV}$ across thresholds: the PA complementary CDF (CCDF) lies strictly left/below $\mathrm{CONV}$ for all $T$.}
    \label{fig:fig_ccdf}
    \vspace{-10pt} 
\end{figure}

\subsection{PA: Bottleneck Distance via $m$-Spacings}
\label{subsec:PA}
Under PA, the straggler offset equals half the shortest window covering $M$ consecutive users:
\begin{equation}\label{eq:straggler_bottle}
x_{\rm str}^{\rm PA}=\frac{L_M}{2},\qquad
L_M:=\min_{1\le i\le K-M+1}\big\{X_{(i+M-1)}-X_{(i)}\big\}.
\end{equation}
These satisfy the deterministic ordering \(
x^{\mathrm{PA}}_{\mathrm{str}} = \frac{L_M}{2} \;\le\; Y_{[M]} = x^{\mathrm{CONV}}_{\mathrm{str}}
\), which ensures that any upper-tail bound for $Y_{[M]}$ applies to $x_{\rm str}^{\rm PA}$. To quantify PA’s intrinsic performance gain, we analyze $L_M$ directly using $m$-spacings with $m:=M-1\ge 1$.

Let $U_{(j)}:=\frac{X_{(j)}+\frac{D}{2}}{D}\in[0,1]$ denote the normalized order statistics. Define the simple spacings as $G_1:=U_{(1)}$, $G_j:=U_{(j)}-U_{(j-1)}$ for $2\le j\le K$, and $G_{K+1}:=1-U_{(K)}$.  By the theory of uniform order statistics \cite{pyke65}, $(G_1,\dots,G_{K+1})\sim\mathrm{Dirichlet}(1,\dots,1)$.
For $i=1,\ldots,K-m$, the $m$-span is defined as:
\[
S_i^{(m)}:=X_{(i+m)}-X_{(i)},\qquad 
\widetilde S_i^{(m)}:=\frac{S_i^{(m)}}{D}=U_{(i+m)}-U_{(i)},
\]
which further equals $\sum_{j=i+1}^{i+m}G_j$.
By Dirichlet aggregation properties, each fixed $m$-span follows:
\begin{equation}
\widetilde S_i^{(m)}\sim{\rm Beta}\big(m,K{+}1{-}m\big),\qquad 
\mathbb E[\widetilde S_i^{(m)}]=\frac{m}{K+1},
\label{eq:span-moments}
\end{equation}
with second moment $\mathbb E\!\big[(\widetilde S_i^{(m)})^{2}\big]=\frac{m(m+1)}{(K+1)(K+2)}$.

Let $L_M:=\min_{1\le i\le K-m}S_i^{(m)}$ and $\widetilde L_M:=\frac{L_M}{D}$, we have:
\begin{equation}
x_{\rm str}^{\rm PA}= \frac{L_M}{2}=\frac{D}{2}\,\widetilde L_M.
\label{eq:xstr-PA}
\end{equation}

Each interior spacing $G_j$ ($2\le j\le K$) appears in at most $m$ distinct $m$-spans. Summing all $m$-spans with multiplicities $0\le c_j\le m$ yields:
\[
\sum_{i=1}^{K-m}\widetilde S_i^{(m)}=\sum_{j=2}^{K} c_j G_j
\le m\sum_{j=2}^{K}G_j
= m\bigl(1-(G_1+G_{K+1})\bigr).
\]
Therefore,
\[
\widetilde L_M\le \frac{1}{K-m}\sum_{i=1}^{K-m}\widetilde S_i^{(m)}
\le \frac{m}{K-m}\quad\text{a.s.},
\]
leading to the upper bound:
\begin{equation}
\mathbb E\!\big[x_{\rm str}^{\rm PA\,2}\big]
=\frac{D^2}{4}\,\mathbb E[\widetilde L_M^{2}]
\ \le\ \frac{D^2}{4}\left(\frac{m}{K-m}\right)^{\!2}.
\label{eq:PA-UB-avg}
\end{equation}
Using \eqref{eq:span-moments} and the positive inequality $\widetilde L_M\le \widetilde S_i^{(m)}$ also yields:
\begin{equation}
\mathbb E\!\big[x_{\rm str}^{\rm PA\,2}\big]
\le \frac{D^2}{4}\,\frac{m(m+1)}{(K+1)(K+2)}.
\label{eq:PA-UB-beta}
\end{equation}
A uniformly valid upper bound over $(K,M)$ is $\min\{\text{\eqref{eq:PA-UB-avg}},\,\text{\eqref{eq:PA-UB-beta}}\}$.

\begin{proposition}[PA lower bound via the minimum simple spacing]
\label{prop:PA-LB}
Let $m:=M-1\ge1$. Since each $m$-span is a sum of $m$ positive spacings, $\widetilde L_M \ge m\min_{1\le j\le K+1}G_j$. Consequently,
\begin{equation}\label{eq:PA-LB}
\mathbb E\!\big[x_{\rm str}^{\rm PA\,2}\big]
\ \ge\ \frac{D^2}{4}\,m^2\,\frac{2}{(K+1)^3(K+2)}.
\end{equation}
\begin{IEEEproof}
Refer to Appendix~\ref{app:PA_LB_proof}.
\end{IEEEproof}
\end{proposition}

Combining \eqref{eq:conv-YM-second-moment} with \eqref{eq:PA-UB-beta} and using instead \eqref{eq:PA-UB-avg} gives the sharper large‑$K$ limit
\begin{equation}
\limsup_{K\to\infty}\frac{\mathbb E\!\big[x_{\rm str}^{\rm PA\,2}\big]}{\mathbb E\!\big[Y_{[M]}^{2}\big]}
\ \le\ \frac{m^2}{M(M+1)}=\frac{(M-1)^2}{M(M+1)}\ <\ 1.
\label{eq:ratio-const-sharp}
\end{equation}
The $\mathrm{CONV}$ approach exhibits:
\begin{equation}
\mathbb E\!\big[Y_{[M]}^{2}\big]
= \Big(\tfrac{D}{2}\Big)^{2}\frac{M(M+1)}{(K+1)(K+2)}
= \Theta(K^{-2}),
\label{eq:conv-scaling}
\end{equation}
while \eqref{eq:PA-UB-beta} and \eqref{eq:PA-LB} imply that PA achieves:
\begin{equation}
\mathbb E\!\big[x_{\rm str}^{\rm PA\,2}\big]=O(K^{-2})
\quad\text{and}\quad 
\mathbb E\!\big[x_{\rm str}^{\rm PA\,2}\big]=\Omega(K^{-4}).
\label{eq:PA-scaling-bounds}
\end{equation}
See Appendix~\ref{app:PA_consequences} for a compact derivation.
This analysis demonstrates that PA achieves a strict constant-factor improvement in squared straggler distance and can exhibit superior scaling with $K$ compared to $\mathrm{CONV}$.

\begin{remark}
By turning the worst user–radiator offset from the $M$-th closest absolute position $Y_{[M]}$ into half an $m$-spacing $\frac{L_M}{2}$, PA reduces the critical path in SFL deterministically and raises per‑deadline feasibility in AFL. The improvement is scheduler‑agnostic and grows with $M$. 
\end{remark}

\section{Latency Advantages of Pinching Antennas}
\label{sec:latency}

Section~\ref{subsec:PA} established that in SFL, the PA’s bottleneck distance is reduced from $Y_{[M]}$ (for $\mathrm{CONV}$) to $\frac{L_M}{2}$ (for PASS), with $\frac{L_M}{2}\le Y_{[M]}$ and
corresponding moment bounds. This section quantifies the resulting latency improvement.

In SFL with $M$ users per round, the bottleneck offsets are
\(
x_{\rm CONV}^{\rm str}=Y_{[M]}\) and \(x_{\rm PA}^{\rm str}=\frac{L_M}{2}\) as derived in Section~\ref{subsec:PA}.

\subsection{SFL latency: PA dominates at all SNRs}
\label{subsec:SFL}

By \eqref{eq:straggler_bottle}, $\frac{L_M}{2}\le Y_{[M]}$ holds deterministically. Since $r\mapsto R(r)$ is strictly decreasing, this ordering is transferred directly to communication rates and latencies.

\begin{theorem}
\label{thm:SFL-dominance}
For any spatial realization,
$T^{\mathrm{SFL}}_{\mathrm{PA}}\ \le\ T^{\mathrm{SFL}}_{\mathrm{CONV}}$.
If $X$ has a continuous distribution, the inequalities are strict almost surely.
\end{theorem}

\subsection{AFL latency: strictly positive PA gain with an explicit SNR threshold}
\label{subsec:AFL}

Define
$\Lambda:=\log_2 \mathcal{S}$,
\(
v(x):=\frac{x^2+d^2}{\mathcal{S}}=2^{-\Lambda}(x^2+d^2),\quad
\delta(x):=-\log_2(x^2{+}d^2)+\frac{1}{\ln 2}\,\ln\!\bigl(1+v(x)\bigr).
\)
Then
\[
R_0(x)=\Lambda-\log_2(x^2{+}d^2)+\frac{\ln(1+v(x))}{\ln 2}
=:\Lambda+\delta(x).
\]

Under $\mathrm{CONV}$, $x\sim\mathcal U[-\frac{D}{2},\frac{D}{2}]$ and $z=0$, yielding:
\begin{equation}
\label{eq:E-invR-CONV-AFL}
\mathbb{E}\!\Big[\frac{1}{R_{\rm CONV}}\Big]
=\frac{1}{\Lambda}+\frac{\bar\ell_{\rm CONV}}{\Lambda^2}+\mathbb{E}\big[\mathcal R_\Lambda(X)\big],
\end{equation}
where $R_\Lambda(x)$ denote the remainder of the $\frac{1}{\Lambda}$ expansion of $\frac{1}{R_0(x)}$, $\bar\ell_{\rm CONV}:=\frac{1}{D}\int_{-\frac{D}{2}}^{\frac{D}{2}}\log_2(x^2{+}d^2)\,dx$.
The closed-form expression is:
\begin{equation}
\label{eq:ell-CONV-closed}
\bar\ell_{\mathrm{CONV}}
= \log_2(d^2) + \frac{1}{\ln 2}\!\left(\ln(1+\zeta^2) - 2 + \frac{2}{\zeta}\arctan\zeta\right),
\end{equation}
where $\zeta$ denotes $\frac{D}{2d}$.

\noindent Under PA pinning, $x-z=0$, hence $R_{\rm PA}=R_0(0)=\log_2(1+\frac{\mathcal{S}}{d^2})$ and
\begin{equation}
\label{eq:E-invR-PA-AFL}
\mathbb{E}\!\Big[\frac{1}{R_{\rm PA}}\Big]=\frac{1}{R_0(0)}
=\frac{1}{\Lambda}+\frac{\log_2(d^2)}{\Lambda^2}+\mathcal R_\Lambda(0).
\end{equation}
Define
\(
g(\zeta):=\ln(1{+}\zeta^2)-2+\frac{2}{\zeta}\arctan \zeta\). Since $g'(\zeta)=\frac{2}{\zeta^2}\big(\zeta-\arctan\zeta\big)>0$, we have $g(\zeta)>0$ for all $\zeta >0$.

Subtracting \eqref{eq:E-invR-PA-AFL} from \eqref{eq:E-invR-CONV-AFL} and using
\eqref{eq:ell-CONV-closed} gives
\begin{equation}
\label{eq:AFL-gap-exact}
\mathbb{E}\!\Big[\tfrac{1}{R_{\rm CONV}}\Big]-\mathbb{E}\!\Big[\tfrac{1}{R_{\rm PA}}\Big]
=\frac{g(\zeta)}{\Lambda^2\ln 2}\ +\ \Big(\mathbb{E}[\mathcal R_\Lambda(X)]-\mathcal R_\Lambda(0)\Big).
\end{equation}
By Lemma~\ref{lem:uniform-bridge} in Appendix \ref{app:hsnr}, for $\Lambda\ge \Lambda_0$,
\begin{equation}
    \begin{aligned}
&\Big|\mathbb{E}[\mathcal R_\Lambda(X)]-\mathcal R_\Lambda(0)\Big|
\le 2\sup_{x\in[-\frac{D}{2},\frac{D}{2}]}|\mathcal R_\Lambda(x)| \\
& \qquad \qquad \qquad \qquad \qquad \le \frac{4\,(C_0+C_1 2^{-\Lambda})^2}{\Lambda^3}
+\frac{2\,C_1\,2^{-\Lambda}}{\Lambda^2}.
    \end{aligned}
\end{equation}
The AFL latency improvement over $K$ uploads:
\[
\Delta T_{\rm AFL}\ :=\ \frac{K B_t}{\Delta W}\Big(\mathbb{E}\big[\tfrac{1}{R_{\rm CONV}}\big]-\mathbb{E}\big[\tfrac{1}{R_{\rm PA}}\big]\Big),
\]
satisfies the explicit non‑asymptotic lower bound:
\begin{equation}
\label{eq:AFL-gap-lower}
\Delta T_{\rm AFL}\ \ge\ \frac{K B_t}{\Delta W}\left\{\frac{g(\zeta)}{\Lambda^2\ln 2}
-\frac{4\,(C_0+C_1 2^{-\Lambda})^2}{\Lambda^3}
-\frac{2\,C_1\,2^{-\Lambda}}{\Lambda^2}\right\}.
\end{equation}
By choosing the high-SNR threshold:
\begin{equation}
\label{eq:LambdaStar}
\Lambda_{\!*} := \max\!\left\{\Lambda_0,\ \frac{16(C_0{+}C_1)^2\ln 2}{g(\zeta)},\ \log_2\!\Big(\frac{8\,C_1\ln 2}{g(\zeta)}\Big),\ 1\right\}.
\end{equation}
we ensure that for every $\Lambda\ge \Lambda_{\!*}$,
\begin{equation}
\label{eq:AFL-gap-positive}
\Delta T_{\rm AFL}\ \ge\ \frac{K B_t}{\Delta W}\frac{g(\zeta)}{2\,\Lambda^2\ln 2}\ >\ 0.
\end{equation}
\begin{remark}
    In AFL, the explicit lower bound \eqref{eq:AFL-gap-positive} shows a strictly positive time saving above a computable SNR threshold $\Lambda^\star$, with leading term scaling as $\frac{K B_t}{\Delta W}\frac{g(\zeta)}{\Lambda^2\ln 2}$.
\end{remark}

\section{PA Yields More Participants under a Global Deadline}
\label{sec:pa-deadline-final}
This section analyzes system performance under a fixed wall-clock deadline $T_d$ for each communication round. A user successfully participates if and only if its local computation time $T_c$ plus uplink time $\tau(r)$ does not exceed the deadline $T_d$:
\begin{equation}
T_c + \tau\!\big(r(x,z)\big)\ \le\ T_d.
\label{eq:eligibility}
\end{equation}

Under the conventional antenna, a user at position $x$ has link distance \(r=\sqrt{x^2+d^2}\) and achieves successful participation with probability:
\[
p_{\rm conv}(x;T_d)=F_c\!\Big(T_d-\tau(\sqrt{x^2+d^2})\Big).
\]
Under PA tracking, the radiator pins to the user (\(x-z=0\Rightarrow r\equiv d\)), hence
\(
p_{\rm pa}(T_d)=F_c\!\big(T_d-\tau(d)\big).
\)
The expected participant counts are
\begin{equation}
N_{\rm CONV}(T_d)=K\!\int_{\mathbb{R}} f_X(x)\,p_{\rm conv}(x;T_d)\,dx,
\label{eq:N-general}
\end{equation}
where $f_X$ denotes the probability density function of user positions, and $N_{\rm PA}(T_d)=K\,F_c\!\big(T_d-\tau(d)\big)$.

\begin{theorem}
\label{thm:participation-dominance}
For any density \(f_X\), any nondecreasing \(F_c\), and any strictly increasing latency \(\tau(r)\) in the range \(r\),
\begin{equation}
N_{\rm PA}(T_d)\ \ge\ N_{\rm CONV}(T_d)\,
\label{eq:dominance}
\end{equation}
with \emph{strict} inequality whenever \(\mathbb{P}(|X|>0)>0\) and \(F_c\) is not constant on the entire interval:
\[
\mathcal I\ :=\ [\,T_d-\tau(r_{\max}),\ T_d-\tau(d)\,],\quad
r_{\max}:=\operatorname{ess\,sup}\sqrt{X^2+d^2}.
\]
\end{theorem}
\begin{proof}
For every \(x\), since \(\tau(d)\le \tau(\sqrt{x^2+d^2})\) and \(F_c\) is nondecreasing; hence \(F_c(T_d-\tau(d))\ge F_c(T_d-\tau(\sqrt{x^2+d^2}))\) pointwise. Integrating against \(f_X\) and multiplying by $K$ yields the result. Equality holds if and only if either \(\mathbb P(|X|>0)=0\) or \(F_c\) is constant on \(\mathcal I\); otherwise the inequality is strict.
\end{proof}

\vspace{-1em}
\subsection{Uniform corridor: closed form and near‑threshold law}
Consider \(X\sim\mathcal{U}[-\frac{D}{2},\frac{D}{2}]\) with deterministic computation time \(T_c\equiv t_0\). Define the deadline-limited coverage radius $\rho(T_d)$ by:
\begin{equation}
\tau\!\Big(\sqrt{\rho(T_d)^2+d^2}\Big)=T_d-t_0, \qquad T_d\ge t_0.
\label{eq:rhoT-def}
\end{equation}
A conventional user is eligible if and only if \(|x|\le \rho(T_d)\), yielding:
\begin{equation}
N_{\rm CONV}^{\rm (unif)}(T_d)=K \min\!\Big(\frac{2\rho(T_d)}{D},\,1\Big),
\label{eq:Nconv-unif}
\end{equation}
Recall \(c:=\frac{B_t}{\Delta W}\) and solve \eqref{eq:rhoT-def} yields:
\begin{equation}
\rho(T_d)=\sqrt{\,\frac{\mathcal{S}}{2^{\,B_t/((T_d-t_0)\Delta W)}-1}-d^2\,}\,,\quad
T_d\ge t_0+\tau(d).
\label{eq:rho-closed}
\end{equation}
The derivative with respect to deadline is:
\begin{equation}
\frac{d\rho}{dT}
=\frac{\mathcal{S}\,q(T)\,(\ln 2)\,c}{2\,\rho(T)\,\big(q(T)-1\big)^2\,(T-t_0)^2}\ >\ 0,
\label{eq:rho-derivative}
\end{equation}
where $q(T) = 2^{\frac{c}{T-t_0}}$, confirming that $\rho(T)$ increases strictly with $T$.
Define \(T_{\min}:=t_0+\tau(d)\) and \(\Lambda_d:=\log_2(1+\mathcal{S}/d^2)\), a Taylor expansion at \(T_{\min}\) yields the square‑root law:
\begin{equation}
\rho(T)=\kappa\,\sqrt{\,T-T_{\min}\,}+O(T-T_{\min}),
\label{eq:rho-sqrt}
\end{equation}
where $\kappa=\frac{d^2}{\sqrt{\mathcal{S}}}\sqrt{1+\frac{\mathcal{S}}{d^2}}\ \Lambda_d\ \sqrt{\frac{\ln 2}{c}}\ $.
Therefore, \(\rho(T)\) is concave near \(T_{\min}\). Define $\Delta N_{\rm unif}(T_d):=N_{\rm PA}(T_d)-N_{\rm CONV}^{\rm (unif)}(T_d)$. Consequently,
\begin{equation}
\Delta N_{\rm unif}(T_d)=
\begin{cases}
0, & T_d<T_{\min},\\[3pt]
K\left[1-\dfrac{2\rho(T_d)}{D}\right], & T_{\min}\le T_d<T_{\max},\\[8pt]
0, & T_d\ge T_{\max},
\end{cases}
\label{eq:gap-unif}
\end{equation}
with $T_{\max} := t_0+\tau\!\Big(\sqrt{d^2+(\frac{D}{2})^2}\Big)$.
Hence the gap \(\Delta N_{\rm unif}(T)\) is \emph{monotonically decreasing} in \(T\) on \([T_{\min},T_{\max})\) and, because \(\rho\) is concave near \(T_{\min}\), \(\Delta N_{\rm unif}\) is \emph{convex} near \(T_{\min}\).

\subsection{Impact of user distribution: stochastic dominance}
Since \(N_{\rm PA}(T_d)=K\,F_c(T_d-\tau(d))\) is distribution-independent, distributional effects manifest through:
\begin{equation}
N_{\rm CONV}(T_d)=K\,\mathbb E\!\left[F_c\!\big(T_d-\tau(\sqrt{X^2+d^2})\big)\right]
=K\,\mathbb E\![\,\phi(|X|)\,],
\label{eq:Nconv-general-FSD}
\end{equation}
with $\phi(y):=F_c\!\big(T_d-\tau(\sqrt{y^2+d^2})\big)$. The function $\phi$ is nonincreasing in $y$ since $r(y)$ and $\tau(r)$ are increasing while $F_c$ is nondecreasing.

\begin{proposition}
\label{prop:FSD}
If \(|X|_1\succeq_{\rm st}|X|_2\) (first‑order stochastic dominance), then for every deadline $T_d$,
\begin{equation}
\begin{aligned}  
&\Delta N_1(T_d):=N_{\rm PA}(T_d)-N_{\rm CONV}^{(1)}(T_d)\\
& \qquad \qquad \quad \ge\
N_{\rm PA}(T_d)-N_{\rm CONV}^{(2)}(T_d)=:\Delta N_2(T_d).
\end{aligned}
\end{equation}
\end{proposition}
\begin{proof}
Since $\phi$ nonincreasing, first-order stochastic dominance implies \(\mathbb E[\phi(|X|_1)]\le \mathbb E[\phi(|X|_2)]\). Therefore, greater dispersion in user distribution leads to larger PA advantages.
\end{proof}

\subsection{Symmetric Gaussian mixture: explicit formula and tails}
Consider a symmetric Gaussian mixture:
\begin{equation}
f_X(x)=\tfrac{1}{2}\varphi_\sigma(x-\mu)+\tfrac{1}{2}\varphi_\sigma(x+\mu),
\quad
\varphi_\sigma(x)=\frac{1}{\sqrt{2\pi}\sigma}e^{\frac{-x^2}{2\sigma^2}}.
\label{eq:gm}
\end{equation}
Under the hard deadline \(T_c\equiv t_0\), the success probability of $\mathrm{CONV}$ equals the probability mass in \([-\rho(T_d),\rho(T_d)]\):
\begin{equation}
\begin{aligned}
\mathbb P\!\big(|X|\le \rho(T_d)\big)
&=\tfrac{1}{2}\!\left[\Phi\!\Big(\tfrac{\rho(T_d)-\mu}{\sigma}\Big)-\Phi\!\Big(\tfrac{-\rho(T_d)-\mu}{\sigma}\Big)\right]\\
&\quad+\tfrac{1}{2}\!\left[\Phi\!\Big(\tfrac{\rho(T_d)+\mu}{\sigma}\Big)-\Phi\!\Big(\tfrac{-\rho(T_d)+\mu}{\sigma}\Big)\right],
\end{aligned}
\label{eq:Pgm}
\end{equation}
where \(\Phi\) denotes the standard normal cumulative distribution function. The participation advantage is:
\begin{equation}
\Delta N_{\rm gm}(T_d)=K\cdot \mathbf{1}_{\{T_d\ge t_0+\tau(d)\}}-K\cdot \mathbb P\!\big(|X|\le \rho(T_d)\big).
\label{eq:gap-gm}
\end{equation}

\begin{figure}
    \vspace{-6pt} 
    \centering
    \includegraphics[width=0.8\linewidth]{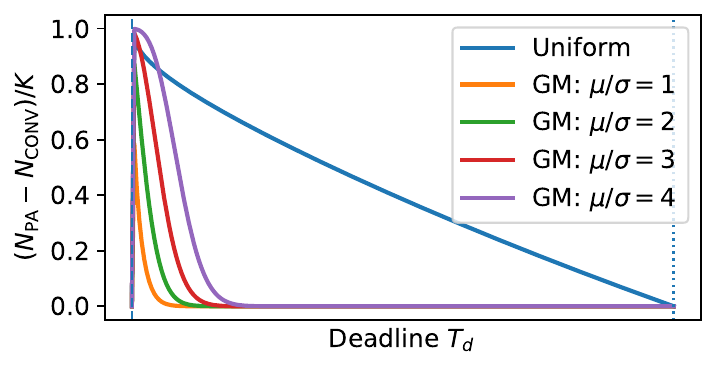}
    \caption{Normalized participation advantage $N_{\rm PA}-N_{\rm CONV})/K$ vs. deadline $T_d$. }
    \label{fig:fig_participant}
    \vspace{-10pt} 
\end{figure}

A sufficient condition for \(|X|_{\rm gm}\succeq_{\rm st}|X|_{\rm unif}\) on \([0,\frac{D}{2}]\) is that the clusters sit sufficiently outside the corridor, e.g. \(\mu-a\sigma\ge \frac{D}{2}\) for some \(a>0\); then \(P_{\rm gm}(|X|>z)\ge 1-2z/D\) for all \(z\in[0,\frac{D}{2}]\). More generally, without asserting FSD, we can compare gaps directly using Mills’ ratio. For \(\rho\in[0,\mu)\),
\begin{equation}
\mathbb{P}_{\rm gm}\!\big(|X|\le \rho\big)
\ \le\
\frac{\sigma}{\sqrt{2\pi}}\left(\frac{e^{-\frac{(\mu-\rho)^2}{2\sigma^2}}}{\mu-\rho}
+\frac{e^{-\frac{(\mu+\rho)^2}{2\sigma^2}}}{\mu+\rho}\right).
\label{eq:Mills-mixture}
\end{equation}
A sufficient condition ensuring larger Gaussian mixture advantages than uniform distribution is:
\begin{equation}
\frac{\sigma}{\sqrt{2\pi}}\!\left(\frac{e^{-\frac{(\mu-\rho(T_d))^2}{2\sigma^2}}}{\mu-\rho(T_d)}
+\frac{e^{-\frac{(\mu+\rho(T_d))^2}{2\sigma^2}}}{\mu+\rho(T_d)}\right)\ \le\ \frac{2\rho(T_d)}{D}.
\label{eq:GM-vs-unif-condition}
\end{equation}
When \(\mu/\sigma\) is large and \(\rho(T_d)\ll \mu\), the left side is exponentially small, so \eqref{eq:GM-vs-unif-condition} holds and \(\Delta N_{\rm gm}(T_d)\ge \Delta N_{\rm unif}(T_d)\).
It is worth mentioning that we focus on uplink; downlink latency is identical across architectures and thus does not affect ordering.

\begin{remark}
    Because $N_{\mathrm{PA}}(T_d)=K\,F_c(T_d-\tau(d))$ is distribution‑free, the entire dependence on user geography sits in $N_{\mathrm{CONV}}(T_d)$ shown in \eqref{eq:Nconv-general-FSD}. More spread only increases the PA–CONV gap, and closed‑forms under uniform corridors reveal a near‑threshold $\sqrt{\,T-T_{\min}}$ law.
\end{remark}
\section{Convergence under PHY‑aware Sampling with Error Feedback}

This section presents a comprehensive convergence analysis for FL systems incorporating physical layer-aware user sampling and error-feedback compression mechanisms. The analysis establishes the theoretical foundation for understanding the interplay between wireless communication constraints and distributed learning performance.

The physical layer enters the FL analysis through two quantities: (i) the minimum inclusion probability and (ii) the compression fidelity. We model them as follows:

\begin{subequations}\label{eq:incprob}
\begin{align}
\pi_{i,t} &= p_s\,F_c\!\big(T_d-\tau_{i,t}\big), \label{eq:incprob_indiv}\\
\pi_{\min,t} &:= \min_{1\le i\le K}\,\pi_{i,t}. \label{eq:incprob_min}
\end{align}
\end{subequations}
\noindent
Here $p_s\in(0,1]$ is the Bernoulli triggering probability, $F_c$ is the CDF of the compute time $T_{c,i}$, and
\(
\tau_{i,t}=\frac{B_t}{W\,R_{i,t}}
\)
is the uplink time with bandwidth $W$ and spectral efficiency $R_{i,t}$.

\begin{subequations}\label{eq:compfid}
\begin{align}
1-\alpha(b_\star) &\;\lesssim\; c_q\,2^{-2b_\star}, \label{eq:compfid_bound}\\
b_t &\equiv b_\star,\quad B_t=d_w\,b_\star\ \ \text{for all } t. \label{eq:compfid_fixed}
\end{align}
\end{subequations}
\noindent
where $\alpha(\cdot)$ denotes the fidelity of the vector quantizer \cite{Vec_quanti_signal_com}, $b_\star\in\mathbb{N}$ is the fixed per-parameter bit budget chosen at design time, $c_q>0$ is the (source/quantizer–dependent) high-rate constant.
The service-rate target $R_{{\rm svc},t}$ is a design-time choice (e.g., $R(d)$ under PA) and does not affect $b_t$.

\subsection{PHY-aware Poisson Sampling}
The sampling mechanism employs a two-stage gating process that accounts for both computational and communication constraints:
\begin{equation}
E_{i,t}:=\mathbf 1\{T_{c,i}\le T_d-\tau_{i,t}\},\quad
Z_{i,t}\sim\mathrm{Bernoulli}(p_s),
\end{equation}
where $E_{i,t}$ denotes the computational feasibility indicator, $Z_{i,t}$ denotes the random participation trigger, and both processes are mutually independent and independent of the gradient computation and compression randomness, conditioned on $\mathcal F_t$.
The actual participation indicator can be defined as:
\begin{equation}
I_{i,t}=E_{i,t}Z_{i,t}. 
\end{equation}
\begin{assumption}
Conditioned on $\cF_t$, the pairs $\{(E_{i,t},Z_{i,t})\}_{i=1}^K$ are mutually independent across $i$, 
and independent of the compression randomness and stochastic gradients. 
Consequently, $\{I_{i,t}\}_{i=1}^K$ are mutually independent with $I_{i,t}\sim\mathrm{Bernoulli}(\pi_{i,t})$ given $\cF_t$.
\end{assumption}

Under the constraint $R_{i,t} \ge R_{\min,t}$,
the Horvitz–Thompson inclusion probability unconditional on $E_{i,t}$, can be expressed as follows:
\begin{equation}
\label{eq:pi-it-correct}
\pi_{i,t}:=\Pr\{I_{i,t}=1\mid\mathcal F_t\}
= p_s\,F_c\!\big(T_d-\tau_{i,t}\big).
\end{equation}

To quantify the impact of user participation, the sampling-noise amplification factor is defined as:
\begin{equation}
    \Xi^{\mathrm{safe}}_t=\frac{K-1+1/\pi_{\min,t}}{K},
\end{equation}
If $\pi_{\min,t}\ge\pi_0>0$, then $\Xi^{\mathrm{safe}}_t\le\frac{K-1+1/\pi_0}{K}$.
Note that $\tau_{i,t}$ is $\cF_t$-measurable and conditioned on $\cF_t$, $T_{c,i}$ is independent of $(\vg_{i,t},\ve_{i,t})$ and of $\tau_{i,t}$. Hence $p_{c,i,t}=\Pr\{E_{i,t}=1\mid\cF_t\}
=\Pr\{T_{c,i}\le T_d-\tau_{i,t}\mid\cF_t\} =F_c\!\big(T_d-\tau_{i,t}\big)$, and $\pi_{i,t}=p_s\,p_{c,i,t}$.

\subsection{Compressor and Error Feedback Mechanism}
\subsubsection{Compressor Properties}
We consider a compression operator $Q_b(\cdot)$ that encodes $v\in \mathbb{R}^d$ using $b$ bits per coordinate. We assume that $Q_b$ satisfies a contractive mean-squared error property, i.e., there is a function $\alpha(b)\in(0,1]$ such that
\begin{equation}\label{EE_aggregated}
\E\!\left[\|\cQ_b(\vv)-\vv\|^2\mid\vv\right]\le (1-\alpha(b))\|\vv\|^2,\quad \alpha(b)\in(0,1].  
\end{equation}
The compression leaves at most a $(1-\alpha(b))$ fraction of the squared norm as error on average. For example, under a high-rate uniform quantizer, one typically has $\alpha(b) \approx 1 - c2^{-2b}$, meaning $1-\alpha(b)$ (the error fraction) scales $2^{-2b}$ \cite{NLP_cite}. This mean‑squared contraction model for $Q_b$ is standard in compressed SGD analyzes.

\subsubsection{Error Feedback Update} 

The error‑feedback recursion is a canonical device to keep compression bias bounded and to make the effective update unbiased in aggregation. 

At every round $t$ and for every user $i$, the EF recursion is executed locally, regardless of the sampling gates $(E_{i,t},Z_{i,t})$. The gates only determine which compressed updates are transmitted and aggregated \cite{error_SGD}:
\begin{align}
\ve_{i,t+1}=\vg_{i_t}+\ve_{i,t}-\cQ_{b_t}(\vg_{i,t}+\ve_{i,t}),
\end{align}
Each user computes a local stochastic gradient every global tick, the EF recursion \eqref{EE_aggregated} is executed locally regardless of gates. The gate only decides transmission, i.e., if $I_{i,t} = 0$, the compressed vector is not sent, but the EF residual updates.

\subsection{Horvitz-Thompson Aggregation Analysis}\label{subsec:HT}
\subsubsection{Aggregation Scheme}
\noindent In each round, the server computes an unbiased aggregate of the received updates using a Horvitz–Thompson estimator to account for non-uniform sampling probabilities. The aggregated gradient estimate is defined as:
\begin{equation}
\label{eq:ht}
\widehat{\vg_t}=\frac{1}{K}\sum_{i=1}^K \frac{I_{i,t}}{\pi_{i,t}}\,Y_{i,t},
\end{equation}
where $Y_{i,t}:=\cQ_{b_t}(\vg_{i,t}+\ve_{i,t})=\vg_{i,t}+\ve_{i,t}-\ve_{i,t+1}$ is the compressed, error-adjusted update from user $i$. 
This estimator $\widehat{\vg_t}$ aims to approximate the true gradient $\nabla F(w_t)$ related to the dynamics of the error feedback residuals.
Conditioned on $\cF_t$, the indicators $\{I_{i,t}\}_{i=1}^K$ are mutually independent with
$I_{i,t}\sim\mathrm{Bernoulli}(\pi_{i,t})$,  where
\[
\pi_{i,t}=p_{c,i,t}\,p_s,\qquad
p_{c,i,t}:=\Pr\{E_{i,t}=1\mid\cF_t\}.
\]
We assume that $\pi_{i,t}$ is known to the server, e.g., via calibrated $F_c$ and measured $\tau_{i,t}$.
The conditional expectation becomes
\begin{equation}
\begin{aligned}
& \E\!\big[\widehat{\vg_t}\mid \cF_t\big]
=\frac{1}{K}\sum_{i=1}^K \E\!\Big[\tfrac{I_{i,t}}{\pi_{i,t}}Y_{i,t}\,\Big|\,\cF_t\Big] =\frac{1}{K}\sum_{i=1}^K \E[Y_{i,t}\mid\cF_t]\\
& \qquad \qquad \quad =\nabla F(\vw_t)+\tilde{\ve_t}-\tilde{\ve}^{\,t+1\mid t}.
\end{aligned}
\label{eq:mean}
\end{equation}
where
\(
\tilde{\ve_t}:=\frac{1}{K}\sum_i \E[\ve_{i,t}\mid\cF_t]
\)
and
\(
\tilde{\ve}^{\,t+1\mid t}:=\frac{1}{K}\sum_i \E[\ve_{i, t+1}\mid\cF_t]
\).
If \(\cQ_b\) were mean‑unbiased or EF residuals vanished, the conditional mean would equal \(\nabla F(\vw_t)\).
The analysis does not require $\cQ_b$ to be mean-unbiased. Given \eqref{eq:ht}, we have
\[
\E[Y_{i,t}\mid\cF_t]=\nabla f_i(\vw_t)+\E[\ve_{i,t}-\ve_{i, t+1}\mid\cF_t].
\]
As seen above, $ \mathbb{E}[Y_{i,t} \mid \mathcal{F}_t]$ naturally includes the error feedback terms. This means even if $Q_b$ is biased, the EF mechanism keeps the bias bounded (we will quantify this as an injection floor in the analysis). If $Q_b$ were mean-unbiased (i.e. $\mathbb{E}[Q_b(v)\mid v]=v$), then we would directly have $\mathbb{E}[\tilde{g}_t|\mathcal{F}_t]=\nabla F(w_t)$. In the case with a possibly biased compressor, \eqref{eq:mean} shows a small bias term $\tilde{e}_t - \tilde{e}^{t+1|t}$ remains.

\subsubsection{Conditional Second Moment}

Under the mutual independence of $\{I_{i,t}\}$ and independence from $\{Y_{i,t}\}$, given $\cF_t$, the exact conditional second moment is
\begin{equation}
\begin{aligned}
&\E\!\left[\|\widehat{\vg_t}\|^2\ \Big|\ \{Y_{i,t}\},\cF_t\right]
=\frac{1}{K^2}\Big\|\sum_{i=1}^K Y_{i,t}\Big\|^2 \\
& \qquad \qquad \qquad \qquad \qquad +\frac{1}{K^2}\sum_{i=1}^K\Big(\frac{1}{\pi_{i,t}}-1\Big)\|Y_{i,t}\|^2.
\end{aligned}
\end{equation}
Applying the inequality $\|\sum_i \mathbf y_i\|^2\le K\sum_i\|\mathbf y_i\|^2$ and taking expectations yields 
\begin{equation}
\E\!\left[\|\widehat{\vg_t}\|^2\mid\cF_t\right]
\le \underbrace{\frac{K-1+1/\pi_{\min,t}}{K}}_{:=\,\Xi^{\rm safe}_t}\cdot \frac{1}{K}\sum_{i=1}^K \E\big[\|Y_{i,t}\|^2\mid\cF_t\big].
\end{equation}
Using $Y_{i,t}=\vg_{i,t}+\ve_{i,t}-\ve_{i, t+1}$ and $\|a+b-c\|^2\le 3(\|a\|^2+\|b\|^2+\|c\|^2)$, we obtain
\begin{equation}\label{eq:HT2}
\begin{aligned}
&\E\!\left[\|\widehat{\vg_t}\|^2\mid\cF_t\right]
\le 3\,\Xi^{\rm safe}_t\Big(\|\nabla F(\vw_t)\|^2+\delta^2+\sigma^2\Big) \\
& \qquad \qquad \qquad \quad +3\,\Xi^{\rm safe}_t\Big(\cE^t_{\rm cond}+\bar{\cE}^{\,t+1\mid t}\Big),
\end{aligned}
\end{equation}
where $\cE^t_{\rm cond}:=\frac{1}{K}\sum_{i=1}^K \E\big[\|\ve_{i,t}\|^2\mid\cF_t\big]$ and $\bar{\cE}^{\,t+1\mid t}:=\frac{1}{K}\sum_{i=1}^K \E\big[\|\ve_{i, t+1}\|^2\mid\cF_t\big]$. 

\begin{assumption}
There exists $G^2<\infty$ such that $\sup_t \E[G_t^2]\le G^2$, where $G_t^2=\frac{1}{K}\sum_{i=1}^K\|\nabla f_i(\vw_t)\|^2$.
\end{assumption}
Since \(\E[\vg_{i,t}\mid\cF_t]=\nabla f_i(\vw_t)\) and \(\E[\|\vg_{i,t}-\nabla f_i(\vw_t)\|^2\mid\cF_t]\le\sigma^2\), we have
\begin{equation}
\E\!\big[\|\vg_{i,t}\|^2\mid\cF_t\big]\le\ \big\|\nabla f_i(\vw_t)\big\|^2+\sigma^2.
\end{equation}
Averaging over all users yields
\begin{equation}
\label{eq:Gt_proxy_bound}
\frac{1}{K}\sum_{i=1}^K \E \big[\|\vg_{i,t}\|^2\mid\cF_t\big]
\ \le\ G_t^2+\sigma^2.
\end{equation}
Moreover, using the heterogeneity bound
\(
\frac{1}{K}\sum_i \|\nabla f_i(\vw)-\nabla F(\vw)\|^2\le \delta^2
\),
then
\(
G_t^2 = \|\nabla F(\vw_t)\|^2\ +\ \frac{1}{K}\sum_{i=1}^K \|\nabla f_i(\vw_t)-\nabla F(\vw_t)\|^2, 
\)
we have
\begin{equation}\label{eq:Gt_vs_global}
G_t^2  \le \|\nabla F(\vw_t)\|^2+\delta^2.
\end{equation}

\subsection{EF Residual Dynamics}

The error feedback mechanism induces a recurrence relation for the residual norms. Applying the mean-squared error bound with $v = g_{i,t} + e_i^t$, we obtain
$$
\E\!\left[\|\ve_{i, t+1}\|^2\mid \cF_t\right]
\le (1-\alpha(b_t))\,\E\!\left[\|v\|^2\mid\cF_t\right].
$$
For any $c>0$, applying the inequality $\|a+b\|^2\le (1+c)\|a\|^2+(1+1/c)\|b\|^2$ with $a=\ve_{i,t}$, $b=\vg_{i,t}$ and taking $\E[\cdot\mid\cF_t]$,

\begin{equation}
\begin{aligned}
&\E\!\left[\|\ve_{i, t+1}\|^2\mid \cF_t\right]
\le \underbrace{(1-\alpha(b_t))(1+c)}_{:=\,\rho(b_t)}\,\E\!\left[\|\ve_{i,t}\|^2\mid\cF_t\right] \\
& \qquad \qquad \qquad \quad+\underbrace{(1-\alpha(b_t))(1+1/c)}_{:=\,c_1(1-\alpha(b_t))}\,\E\!\left[\|\vg_{i,t}\|^2\mid\cF_t\right].
\end{aligned}
\end{equation}
Setting $c:=\alpha(b_t)/2$ ensures $\rho(b_t)=(1-\alpha(b_t))\Bigl(1+\tfrac{\alpha(b_t)}{2}\Bigr)
=1-\tfrac{\alpha(b_t)}{2}-\tfrac{\alpha(b_t)^2}{2} < 1$. 
Averaging over users yields the conditional update
\begin{equation}\label{eq:EF-cond}
\bar{\cE}^{\,t+1\mid t}
\le \rho(b_t)\,\cE^t_{\rm cond}
+ c_1(1-\alpha(b_t))\,\big(G_t^2+\sigma^2\big).    
\end{equation}
Assume that there exists $G^2$ such that \(\sup_t \E[G_t^2]\le G^2\). Taking total expectations and using uniform bounds yields
\begin{equation}
\label{eq:EF-tot}
\cE^{t+1}\le \rho_{\max}\,\cE^t + c_1(1-\alpha)_{\max}\big(G^2+\sigma^2\big),
\end{equation}
where \(\cE^t:=\E[\cE^t_{\rm cond}]\), \(\rho_{\max}:=\sup_t \rho(b_t)<1\),
and \((1-\alpha)_{\max}:=\sup_t (1-\alpha(b_t))\).

\subsection{One‑Step Descent and Lyapunov Construction}
By $L$-smoothness of the objective function, we have
\begin{equation}
\label{eq:onestep}
\begin{aligned}
\E\!\left[F(\vw_{t+1})\mid\cF_t\right]
&\le F(\vw_t)-\eta\,\E\!\left[\langle\nabla F(\vw_t),\widehat{\vg_t}\rangle\mid\cF_t\right]
\\
&\quad +\frac{L\eta^2}{2}\,\E\!\left[\|\widehat{\vg_t}\|^2\mid\cF_t\right].
\end{aligned}
\end{equation}
Using \eqref{eq:mean} and Young’s inequality with parameter \(\beta=L\eta\),
\[
-\eta\langle\nabla F(\vw_t),\tilde{\ve_t}-\tilde{\ve}^{\,t+1\mid t}\rangle
\le \frac{L\eta^2}{2}\|\nabla F(\vw_t)\|^2
+\frac{1}{2L}\|\tilde{\ve_t}-\tilde{\ve}^{\,t+1\mid t}\|^2,
\]
with Jensen and \(\|a-b\|^2\le 2\|a\|^2+2\|b\|^2\) give
\begin{equation}
\label{eq:Deltae}
\|\tilde{\ve_t}-\tilde{\ve}^{\,t+1\mid t}\|^2
\le 2\cE^t_{\rm cond}+2\bar{\cE}^{\,t+1\mid t}.
\end{equation}
Combining \eqref{eq:HT2}–\eqref{eq:Deltae} with \eqref{eq:onestep} yields
\begin{equation}
\label{eq:onestep-raw}
\begin{aligned}
\E[F(\vw_{t+1})\mid\cF_t]
&\le F(\vw_t) - \eta\gamma_t\|\nabla F(\vw_t)\|^2
\\
&\quad + A_t\big(\cE^t_{\rm cond}+\bar{\cE}^{\,t+1\mid t}\big) \\
& \qquad + \frac{3}{2}L\eta^2\Xi^{\rm safe}_t(\sigma^2+\delta^2),
\end{aligned}
\end{equation}
where $A_t:=\tfrac{1}{L}+\tfrac{3}{2}L\eta^2\Xi^{\rm safe}_t$ and $\gamma_t:=1-\tfrac{L\eta}{2}(1+3\Xi^{\rm safe}_t)$.
\subsection{Lyapunov Function Construction}
Define the Lyapunov potential $\Psi^t:=\E[F(\vw_t)]+\lambda\,\cE^t$ with $\lambda>0$. 
Taking total expectations of \eqref{eq:onestep-raw} and using $\E[\bar{\cE}^{\,t+1\mid t}]=\cE^{t+1}$, we obtain
\begin{equation}\label{eq:Lyap}
\begin{aligned}
&\Psi^{t+1} \le \Psi^{t} -\eta\,\E[\gamma_t\|\nabla F(\vw_t)\|^2]
+ \tfrac{3}{2}L\eta^2\E[\Xi^{\rm safe}_t](\sigma^2+\delta^2) \\
& \qquad\quad + \underbrace{\big(-\lambda(1-\rho_{\max})+A^+ (1+\rho_{\max})\big)}_{:=\ \mathsf{Coeff}_{\cE}}\,\cE^t \\
& \qquad \qquad + \underbrace{(\lambda+A^+)\,c_1(1-\alpha)_{\max}(G^2+\sigma^2)}_{\text{EF‑injection floor}},
\end{aligned}
\end{equation}
where $A^+:=\sup_t A_t\le \tfrac{1}{L}+\tfrac{3}{2}L\eta^2\,\Xi^{\rm safe}$.

\vspace{-0.5em}
\subsection{Convergence Conditions}
If the stepsize satisfies
\(
\eta\le \dfrac{1}{L\big(1+3\Xi^{\rm safe}\big)},
\)
then $\gamma_t\ge \tfrac12$, ensuring that the gradient term contributes $-\tfrac{\eta}{2}\E\|\nabla F(\vw_t)\|^2$.
The factor $(1+3\Xi^{\rm safe})$ arises from the second–moment control in \eqref{eq:HT2} and the Young–type splitting of cross terms involving the EF residuals.

With \(\rho_{\max}= \big(1-\alpha(b_\star)\big)\Big(1+\tfrac{\alpha(b_\star)}{2}\Big)\),
the choice
\begin{equation}\label{eq:lambda-cond}
    \lambda\ge \dfrac{A^+(1+\rho_{\max})}{1-\rho_{\max}}
\end{equation}
 is feasible.
Then $\mathsf{Coeff}_{\cE}\le 0$ and the residual energy is absorbed. Under conditions above, the Lyapunov function satisfies

\begin{equation}\label{eq:onestep-final}
\begin{aligned}
&\Psi^{t+1}
\le \Psi^t
-\frac{\eta}{2}\,\E\|\nabla F(\vw_t)\|^2
+\underbrace{\frac{3}{2}L\eta^2\,\Xi^{\rm safe}(\sigma^2+\delta^2)}_{\text{variance floor}} \\
& \qquad \quad +\underbrace{(\lambda+A^+)\,c_1(1-\alpha)_{\max}(G^2+\sigma^2)}_{\text{EF injection floor}}.
\end{aligned}
\end{equation}
Three forces govern each step: 
(i) a descent term $-\tfrac{\eta}{2}\E\|\nabla F(\vw_t)\|^2$; 
(ii) a \emph{variance floor} $\propto \Xi^{\rm safe}$ due to sampling and data heterogeneity ($\sigma^2,\delta^2$); 
(iii) an \emph{EF floor} from finite-bit compression. By enlarging the eligible set (Sec.~\ref{sec:pa-deadline-final}), PA increases $\pi_{\min,t}$, thus \emph{reduces} $\Xi^{\rm safe}$ and both floors, accelerating convergence.

\subsection{Convergence Rates and Wall-Clock Time}
Averaging \eqref{eq:onestep-final} over \(t=0,\ldots,T-1\) yields
\begin{equation}
\begin{aligned}
&\frac{1}{T}\sum_{t=0}^{T-1}\E\|\nabla F(\vw_t)\|^2
\le \frac{2(\Psi^0-\Psi^T)}{\eta T}
+3L\eta\,\Xi^{\rm safe}(\sigma^2+\delta^2) \\
&\qquad \qquad \qquad\qquad\qquad+\frac{2}{\eta}\,(\lambda+A^+)\,c_1(1-\alpha)_{\max}(G^2+\sigma^2).
\end{aligned}
\end{equation}

Under the \(\mu\)‑PL condition \(\frac{1}{2}\|\nabla F(\vw)\|^2\ge \mu(F(\vw)-F^\star)\), \eqref{eq:onestep-final} implies,

\begin{equation}
\begin{aligned}
&\Psi^{t+1}-\Psi^\star
\le (1-\eta\mu)\,(\Psi^t-\Psi^\star)\\
&  \qquad \qquad \quad+\underbrace{\big(-\lambda(1-\rho_{\max})+A^+(1+\rho_{\max})+\eta\mu \lambda\big)}_{:=\,\mathsf{Coeff}^{\rm(PL)}_{\cE}}\,\cE^t\\
& \qquad \qquad \qquad +\text{(variance floor)}+\text{(EF floor)}.
\end{aligned}
\end{equation}
To ensure $\mathsf{Coeff}^{\rm(PL)}_{\cE}\le 0$, we require
$$
\lambda\big((1-\rho_{\max})-\eta\mu\big)\ \ge\ A^+(1+\rho_{\max})\Longrightarrow
\lambda \ge \frac{A^+(1+\rho_{\max})}{1-\rho_{\max}-\eta\mu},
$$
assuming $\eta\mu<1-\rho_{\max}$. This yields the convergence guarantee
\begin{equation}\label{eq:pl_linear}
\begin{aligned}
&\E[F(\vw_t)-F^\star]+\lambda\,\cE^t \le
(1-\eta\mu)^t\big(\Psi^0-\Psi^\star\big) \\
& \qquad \qquad \qquad \qquad \qquad+\frac{3L\eta}{2\mu}\,\Xi^{\rm safe}(\sigma^2+\delta^2) \\
&\qquad \qquad \qquad \qquad \qquad \quad+\frac{\lambda+A^+}{\eta\mu}\,c_1(1-\alpha)_{\max}(G^2+\sigma^2).
\end{aligned}
\end{equation}
with iteration complexity
\(T_\varepsilon \le \frac{1}{\mu\eta}\log\frac{\Psi^0-\Psi^\star}{\varepsilon}\).

\begin{corollary}
\label{cor:participation}
If a PHY‑aware policy ensures \(\pi'_{i,t}\ge \pi_{i,t}\) for all \(i\) when switching from $\mathrm{CONV}$ to PA (eligible set nondecreasing), then
\(
\Xi^{\text{safe}}_{\text{PA},t} \le \Xi^{\text{safe}}_{\mathrm{CONV},t},
\)
thereby reducing both the variance floor and EF-injection floor in \eqref{eq:onestep-final} and \eqref{eq:pl_linear}.
\end{corollary}

With $\Xi_{\text{safe}}(x)=\frac{K-1+1/x}{K}$ we have $\Xi_{\text{safe}}’(x)=-\frac{1}{Kx^{2}}<0$. Thus, whenever PA raises the minimum inclusion probability $\pi_{\min,t}$ (enabled by the larger eligible set in Theorem \ref{thm:participation-dominance}).

\begin{corollary}
\label{cor:staleness}
Assume \(T_{\text{round}}(t)=T_{\text{comp}}+\frac{d_w b_t}{W R_{\text{svc},t}}\) with the same \(\{T_{\text{comp}},d_w,b_t,W\}\) under $\mathrm{CONV}$ and PA. Let $\Delta_{\max}$ denote the maximum update staleness.
If \(R_{\text{svc},t}^{\text{PA}}\ge R_{\text{svc},t}^{\mathrm{CONV}}\) for all \(t\), then
\(T_{\text{round}}^{\text{PA}}(t)\le T_{\text{round}}^{\mathrm{CONV}}(t)\) and, for a fixed wall‑clock horizon,
\(\Delta_{\max}^{\text{PA}}\le \Delta_{\max}^{\mathrm{CONV}}\). A sufficient condition is $\eta\le \frac{c_0}{L(1+\Delta_{\max})}$ for some absolute constant $c_0\in(0,1]$ \cite{ASGD}, and consequently, 
\(
\eta_{\max}^{\text{PA}} \ge \eta_{\max}^{\mathrm{CONV}},
\)
so PA allows a larger admissible step size and improves wall-clock convergence.
\end{corollary}

\begin{figure}[t]
  \vspace{-10pt} 
  \centering
    \vspace{-6pt} 
  \begin{subfigure}{0.47\textwidth}
    \centering
    \includegraphics[width=\linewidth]{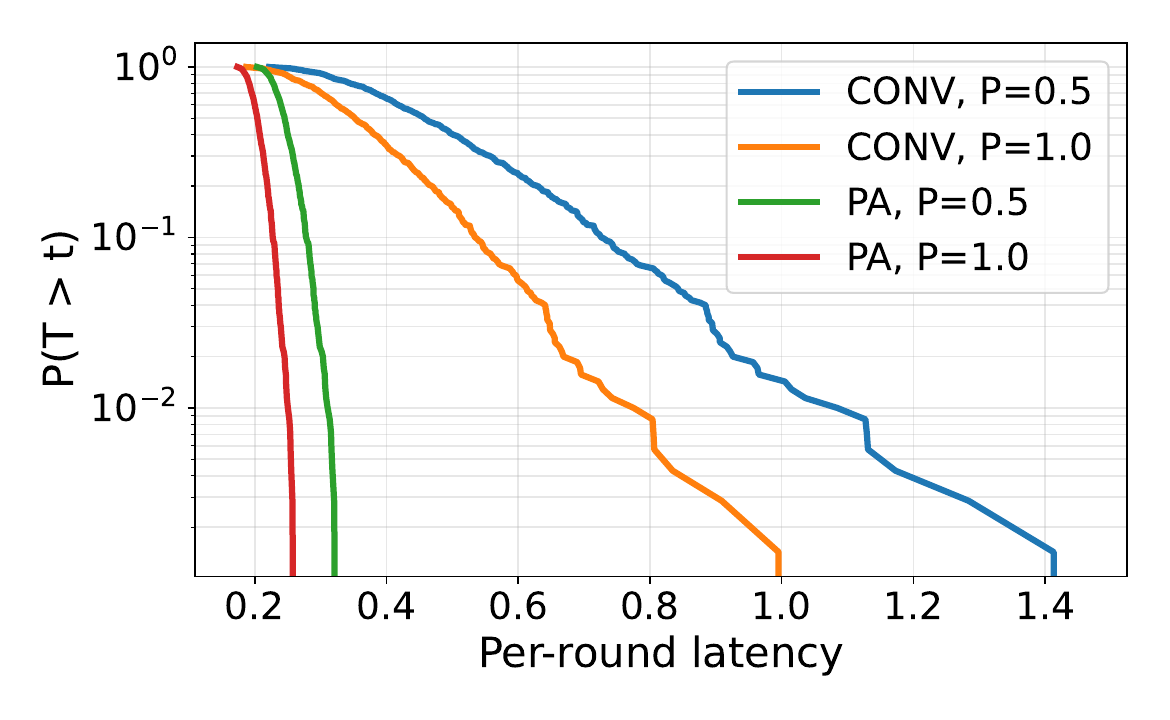}
    \subcaption{SFL: CCDF of per-round latency \(T\).}
    \label{fig:ccdf-sfl}
  \end{subfigure}\hfill
  \begin{subfigure}{0.47\textwidth}
    \centering
    \includegraphics[width=\linewidth]{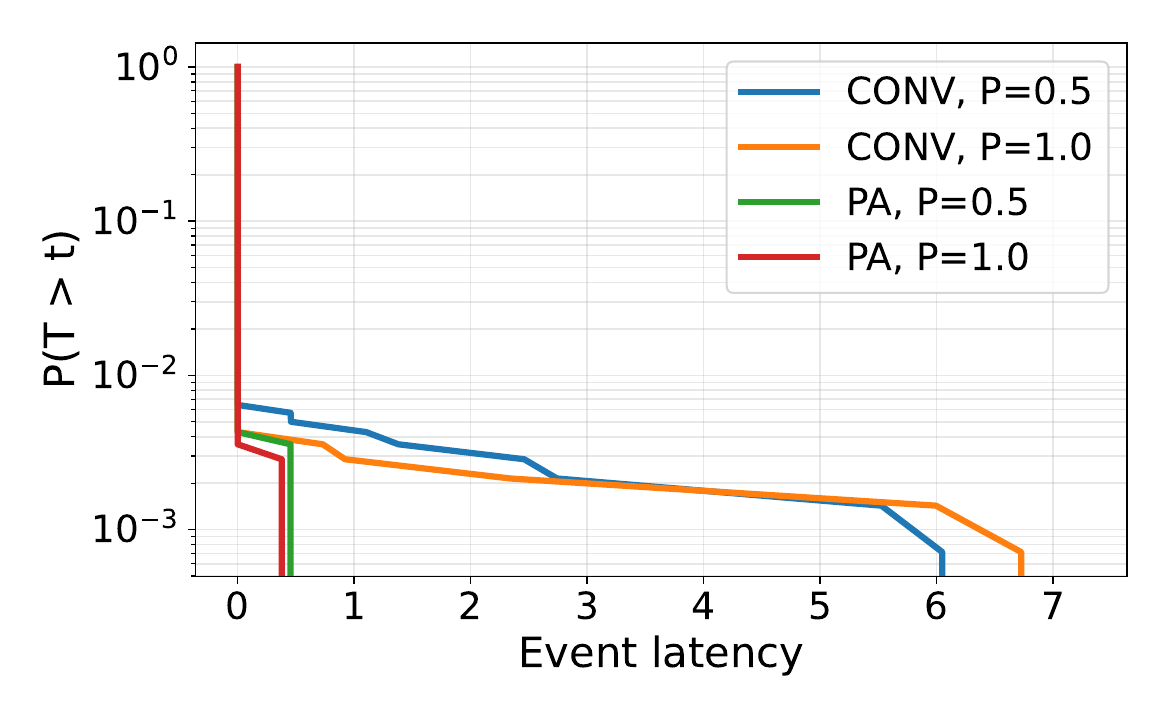}
    \subcaption{AFL: CCDF of per-event (upload) latency \(T\).}
    \label{fig:ccdf-afl}
  \end{subfigure}
  \caption{CCDF \( \Pr[T>t] \) under $\mathrm{CONV}$ and PASS with $D=10$ m, $d=3$ m and $W=1$ MHz. PASS shifts the distribution left and steepens the right tail in both SFL and AFL.}
  \label{fig:ccdf}
    \vspace{-10pt} 
\end{figure}

\section{Simulation Results}
In this section, the simulation results are presented to evaluate the impact of the PASS on FL under both the SFL and AFL.
We consider a wireless FL scenario with $K=40$ users uniformly distributed with $W=1\,\mathrm{MHz}$, $D=10$ m and $d =3$ m \cite{Ding_PASS_Original}.
For demonstration purposes, we conduct the simulation on MNIST with a simple two-layer MLP. Unless otherwise stated, in each SFL round, the server selects $M$ users, and in AFL, we assume a periodic trigger to which each user responds with probability. \(B_t\) are kept identical between $\mathrm{CONV}$ and PASS for fairness unless we explicitly vary the compression level. For compression/EF, symmetric uniform quantization with error-feedback is used with bit-width $b\in\{4,6,8\}$. The IID data distribution is considered across all users. For fair comparisons, each simulation related to FL is optimized over learning rates in two orders of magnitude with mini-batch size $64$ (test batch $256$), i.e., $\eta \in \{10^{-3},5\times10^{-3},\cdots,10^{-1},5\times10^{-1}\}$, with each learning rate averaged over 3 runs.

\begin{figure}[t]
  \vspace{-10pt} 
  \centering
  \begin{subfigure}{0.47\textwidth}
    \centering
    \includegraphics[width=\linewidth]{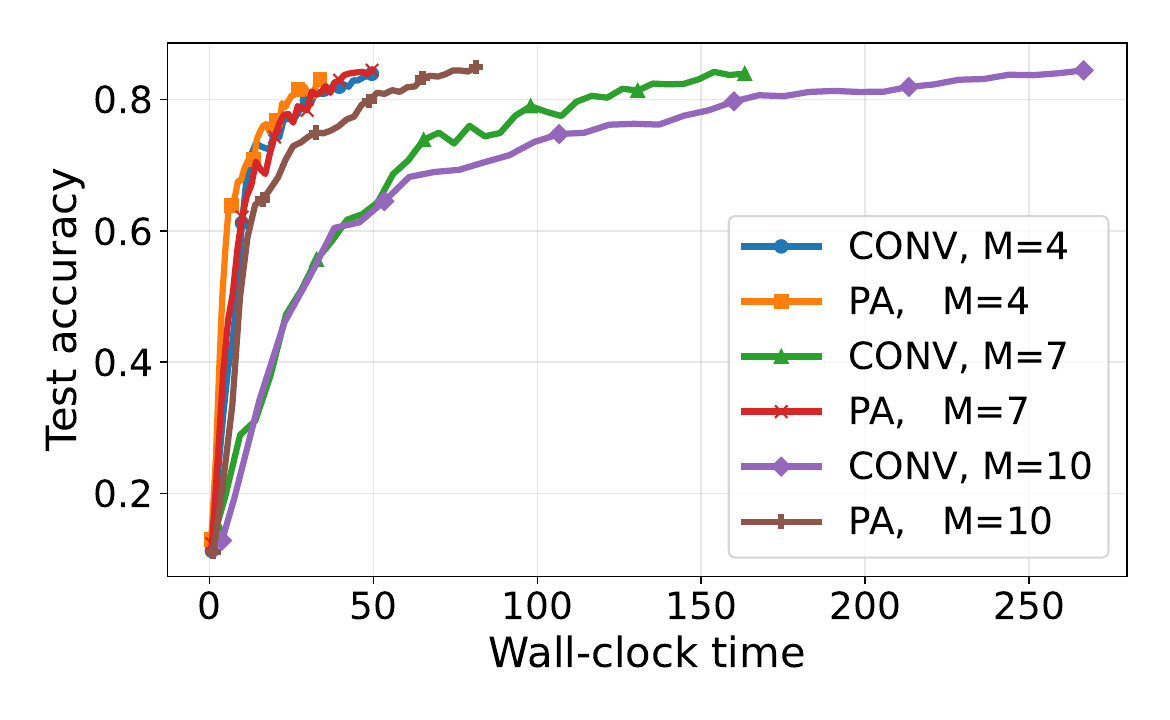}
    \subcaption{SFL: varying scheduled users \(M\in\{4,7,10\}\).}
    \label{fig:acc-sfl}
  \end{subfigure}\hfill
  \begin{subfigure}{0.47\textwidth}
    \centering
    \includegraphics[width=\linewidth]{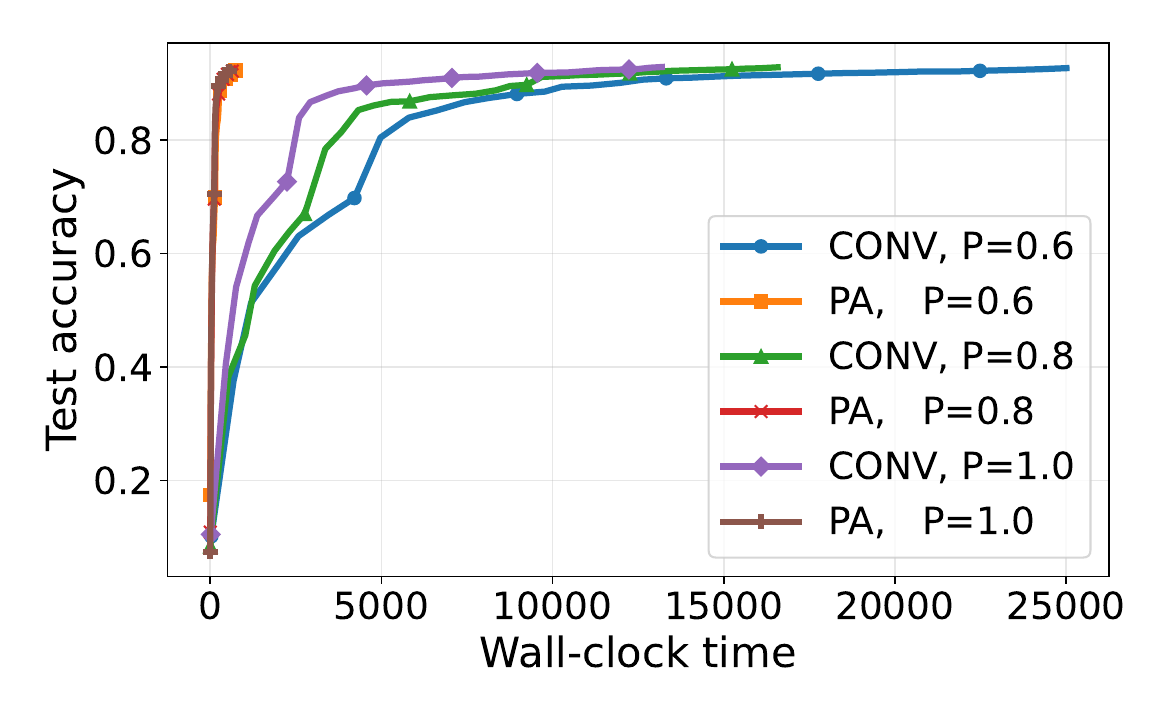}
    \subcaption{AFL: varying power \(P\in\{0.6,0.8,1.0\}\).}
    \label{fig:acc-afl}
  \end{subfigure}
  \caption{Test accuracy versus wall-clock time under $\mathrm{CONV}$ and PASS. PASS shortens links and per-event/round latency, yielding faster wall-clock convergence in both SFL and AFL.}
  \label{fig:acc}
    \vspace{-10pt} 
\end{figure}
It can be seen that PA markedly shifts the distribution left and steepens the right–tail decay in Fig. \ref{fig:ccdf}. For a fixed $t$, the exceedance probability under PA is lower by well over an order of magnitude across most of the operating range. Increasing power $P$ from 0.5 to 1 W further shortens the tail under both architectures, but the dominant effect is the PA-induced reduction of worst-link latency.
In AFL, the CCDF collapses near the minimum latency under PA. By contrast, $\mathrm{CONV}$ exhibits heavy multisecond tails even at $P=1$ W. This indicates a regime shift: PA lifts link rates and suppresses latency dispersion so strongly that AFL becomes compute/trigger-limited rather than communication-limited, rendering participation probability a secondary factor in wall-clock progress. The empirical CCDFs of the figure illustrate this elimination of the tail and the resulting insensitivity to $P$. 
These tail reductions increase $\pi_{\min}$ and shrink $\Xi_{\mathrm{safe}}$, via the stepsize condition and the floors.

Fig. \ref{fig:acc-sfl} studies the impact of test accuracy on the wall-clock time under SFL, comparing $\mathrm{CONV}$ and the proposed PA.
Across all $M$, PA achieves a target accuracy in less time than $\mathrm{CONV}$, with the separation widening as $M$ grows. This scaling is consistent with the SFL critical-path being the maximum uplink time among the $M$ selected users; by pinning the radiator to the scheduled set’s midpoint per round, PA shortens the longest path (reduces the latency tail), shrinking per-round duration and increasing the number of rounds completed per unit time. Final accuracies are similar; the dominant effect is reduced wall-clock convergence. 
This aligns with a larger admissible stepsize under PASS, in addition to shorter rounds.
Fig. \ref{fig:acc-afl} illustrates the test accuracy versus wall-clock time under asynchronous FL for power $P\in\{0.6,0.8,1.0\}$.
Surprisingly, PA drives a near-vertical rise in accuracy, reaching high accuracy in a small fraction of the time required by $\mathrm{CONV}$. This surprising robustness to sparser participation indicates that in AFL progress is limited primarily by how quickly uploads finish, rather than by the participation level itself. By increasing the effective spectral efficiency of the users and sharply contracting the latency distribution, PA accelerates the stream of updates so strongly that the lower $P$ ceases to be the bottleneck. In contrast, $\mathrm{CONV}$ remains event-latency-limited, so reduced $P$ translates into markedly slower wall-clock convergence. Together with reduced staleness in Corollary \ref{cor:staleness}, PASS admits a larger stable stepsize and lower floors, yielding the observed wall-clock speedup even at smaller $P$.

\begin{figure}[t]
  \vspace{-10pt} 
  \centering
  \begin{subfigure}{0.47\textwidth}
    \centering
    \includegraphics[width=\linewidth]{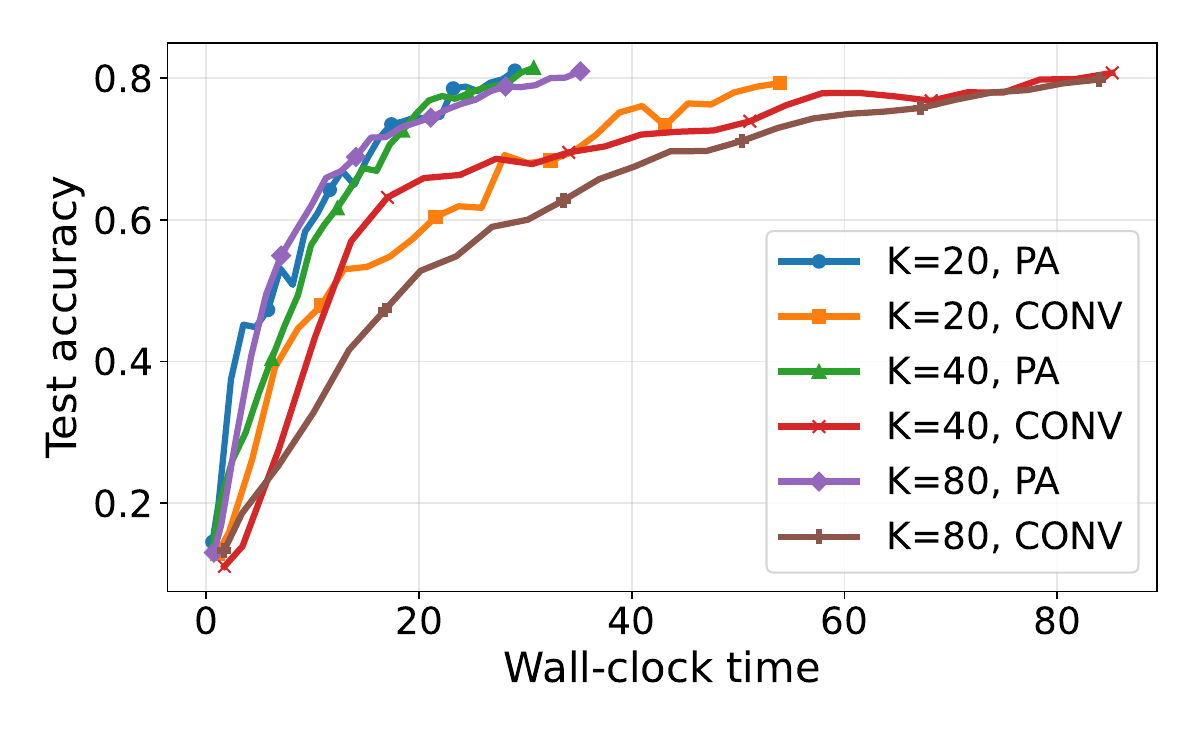}
    \subcaption{SFL: scaling with the number of users $K$.}
    \label{fig:acc-sfl_k}
  \end{subfigure}\hfill
  \begin{subfigure}{0.47\textwidth}
    \centering
    \includegraphics[width=\linewidth]{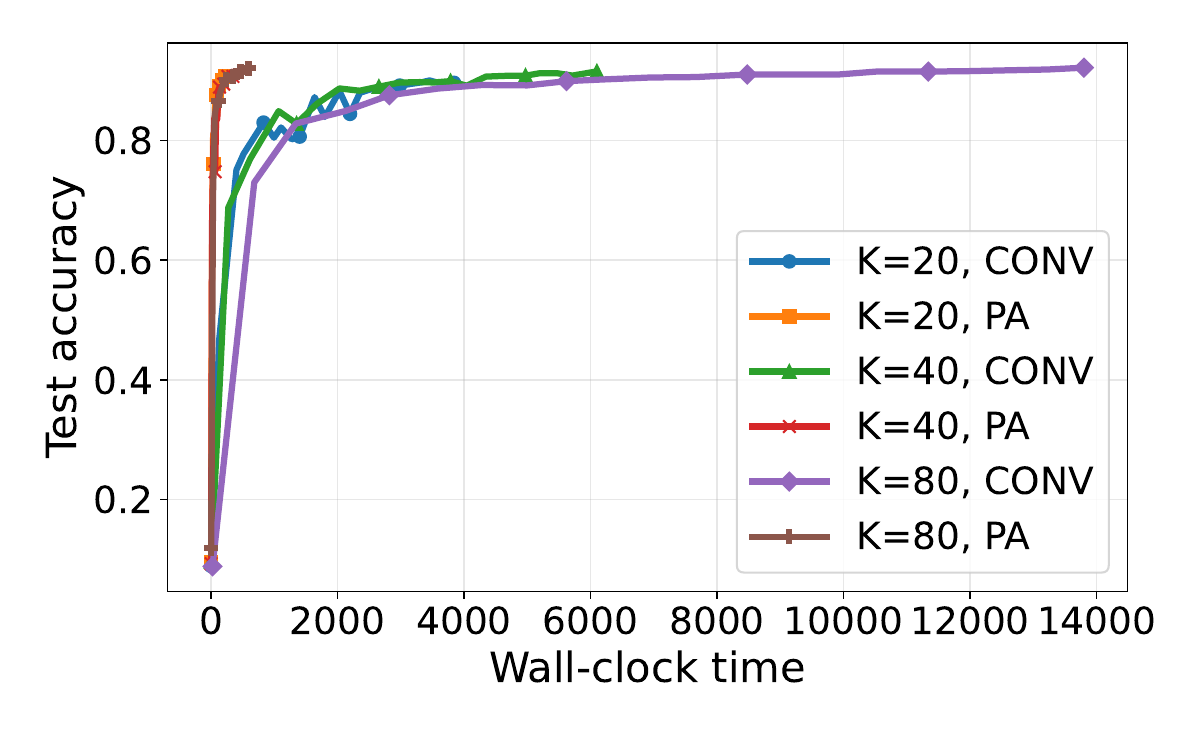}
    \subcaption{AFL: scaling with the number of users $K$.}
    \label{fig:acc-afl_K}
  \end{subfigure}
  \caption{Test accuracy vs. wall-clock time in FL as $K$ scales. }
  \label{fig:acc_k}
    \vspace{-0.8em} 
\end{figure}

Fig. \ref{fig:acc_k} illustrates that PASS consistently reaches a target accuracy sooner, and the gap widens with more users because pinching shortens the round’s critical path (the slowest uplink), reducing per-round duration and increasing rounds completed per unit time. In the AFL case shown in Fig. \ref{fig:acc-afl_K}, the curves rise rapidly and are nearly invariant to user count, indicating a shift toward compute/trigger-limited behavior under PASS. In contrast, $\mathrm{CONV}$ slows markedly as user count grows due to heavy straggler tails, making wall-clock progress sensitive to participation and latency dispersion.

\begin{figure}[htp]
  \vspace{-10pt} 
  \centering
  \begin{subfigure}{0.47\textwidth}
    \centering
    \includegraphics[width=\linewidth]{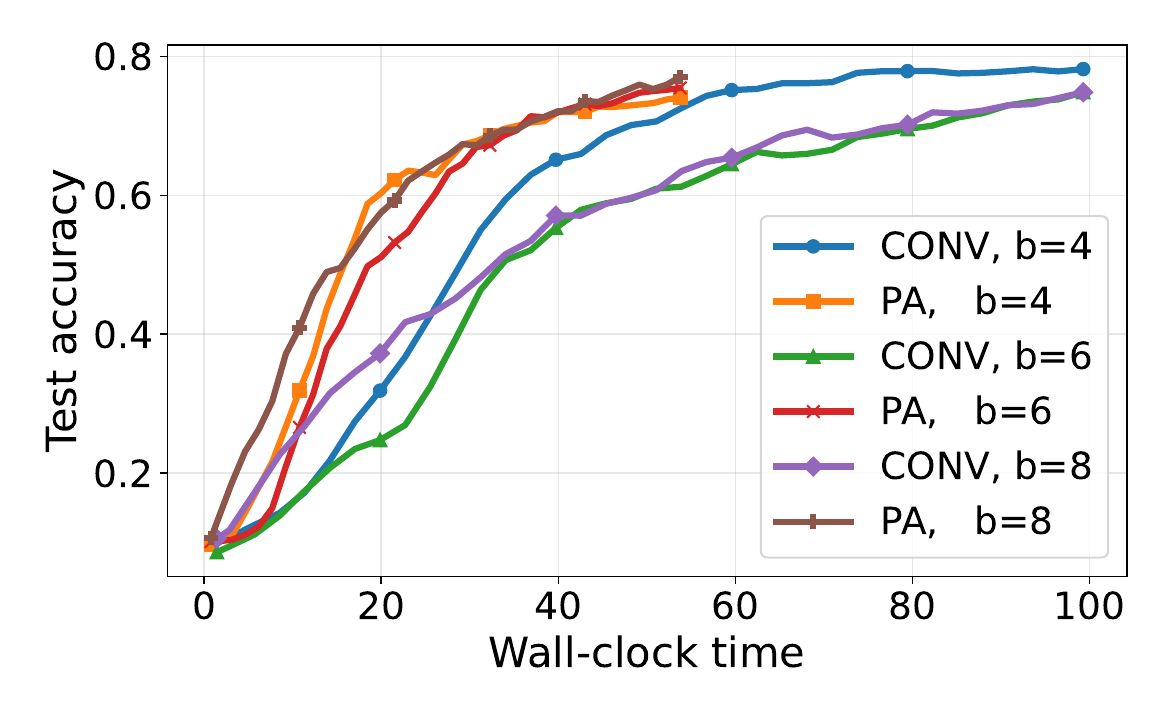}
    \subcaption{SFL: varying compression bits \(b\in\{4,6,8\}\).}
    \label{fig:acc-b-sfl}
  \end{subfigure}\hfill
  \begin{subfigure}{0.47\textwidth}
    \centering
    \includegraphics[width=\linewidth]{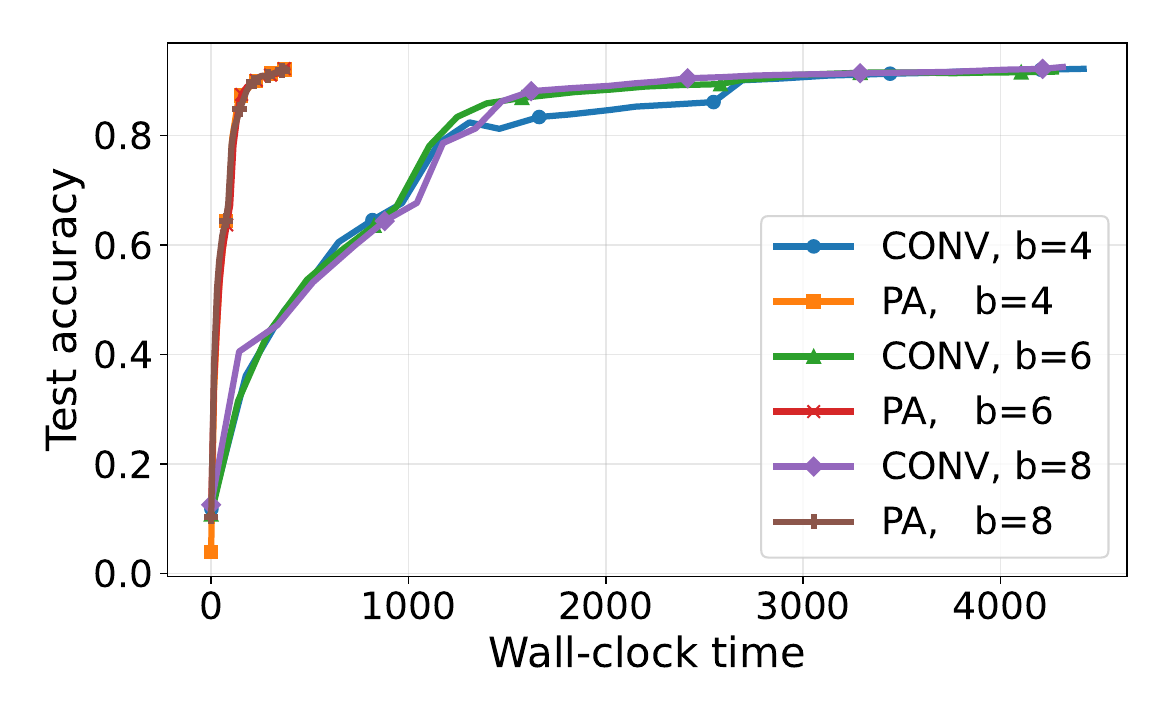}
    \subcaption{AFL: varying compression bits \(b\in\{4,6,8\}\).}
    \label{fig:acc-b-afl}
  \end{subfigure}
  \caption{Effect of payload size via quantization bits \(b\): PASS vs.\ $\mathrm{CONV}$.}
  \label{fig:acc-b}
    \vspace{-9pt} 
\end{figure}

Fig.~\ref{fig:acc-b-sfl} shows that PASS consistently reduces time-to-accuracy for all \(b\), with a modestly widening gap as \(b\) increases. This is due to the SFL maximum-uplink bottleneck: Although error feedback compression reduces the payload of each user linearly in \(b\), the worst link still governs each round. PASS mitigates this tail even for small payloads and the benefit persists as \(b\) grows. 
In AFL (Fig.~\ref{fig:acc-b-afl}), the PASS curves for \(b=4,6,8\) nearly coincide and reach high accuracy rapidly, whereas $\mathrm{CONV}$ shows a clear ordering with \(b\) (larger \(b \Rightarrow\) slower), reflecting that under PASS the system becomes throughput/compute-limited, while $\mathrm{CONV}$ remains communication-dominated.

\begin{figure}[t]
  \vspace{-10pt} 
  \centering
  \begin{subfigure}{0.47\textwidth}
    \centering
    \includegraphics[width=\linewidth]{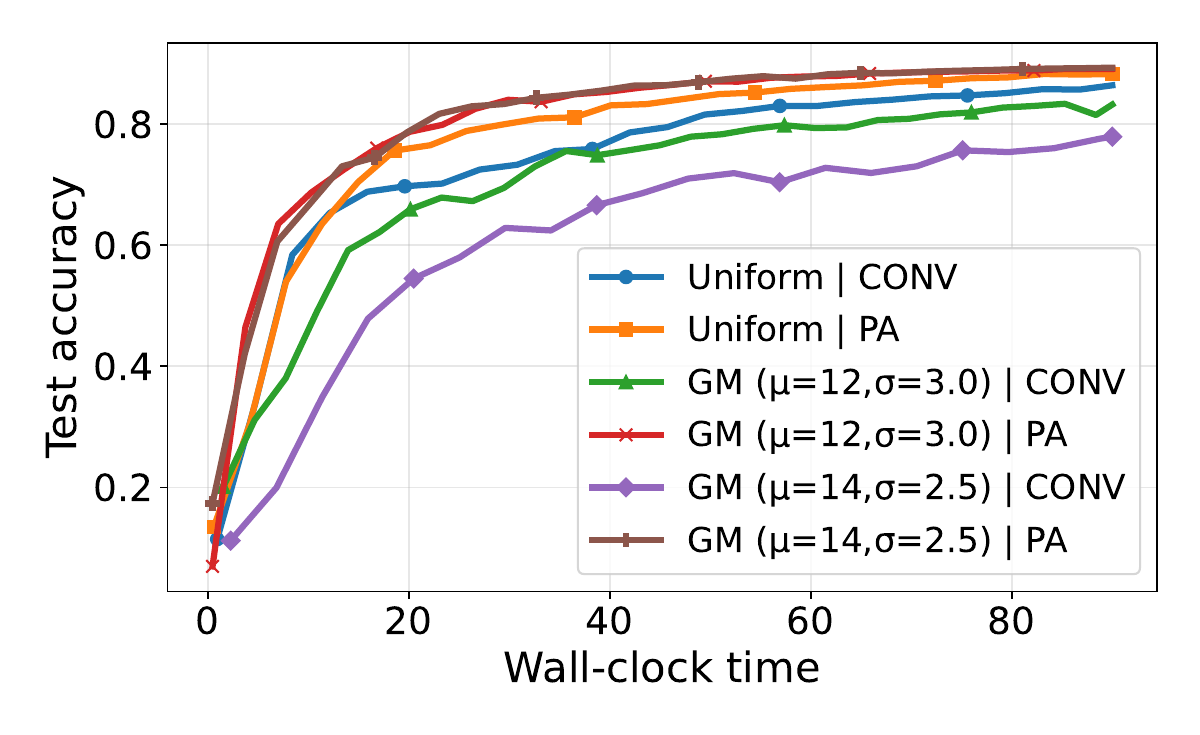}
    \subcaption{SFL: Uniform vs.\ two-component GM.}
    \label{fig:spatial-sfl}
  \end{subfigure}\hfill
  \begin{subfigure}{0.47\textwidth}
    \centering
    \includegraphics[width=\linewidth]{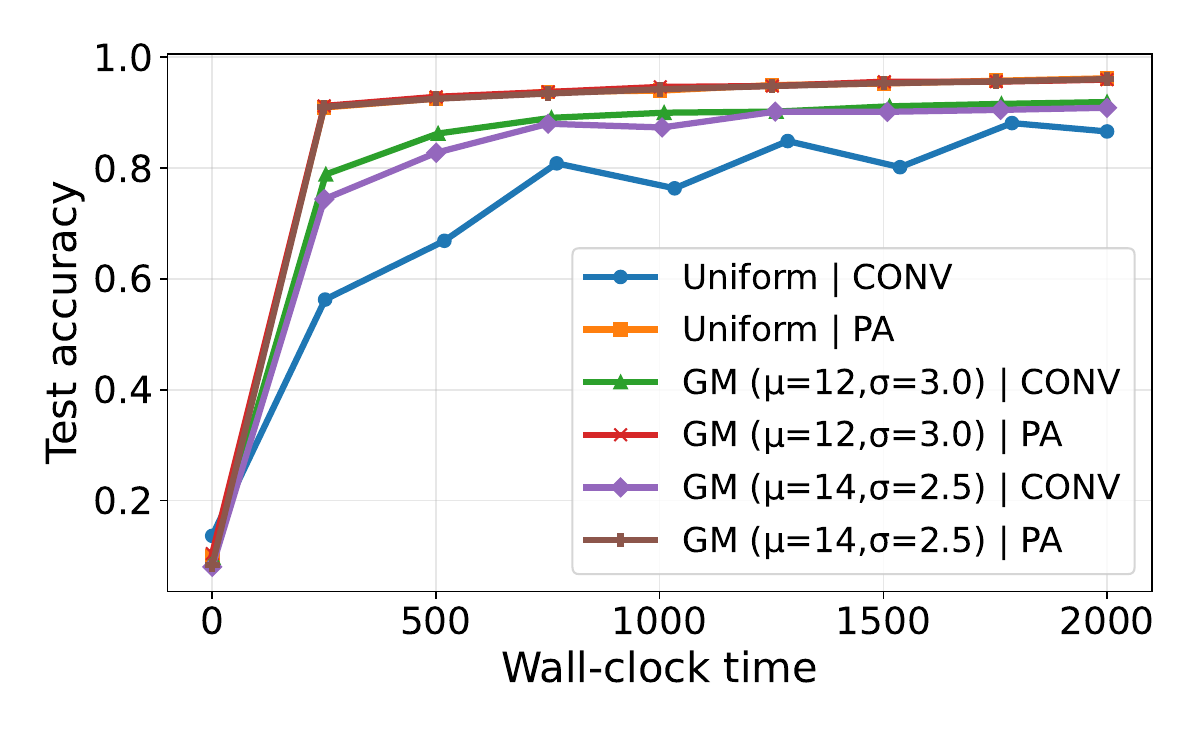}
    \subcaption{AFL: Uniform vs.\ two-component GM.}
    \label{fig:spatial-afl}
  \end{subfigure}
  \caption{Sensitivity to spatial user distributions under $\mathrm{CONV}$ and PASS.}
  \label{fig:spatial}
\end{figure}

Fig.~\ref{fig:spatial-sfl} illustrates test accuracy versus wall-clock time under two spatial distributions, including Uniform and GM with \((\mu,\sigma)\in\{(12,3.0),(14,2.5)\}\). In all cases, PASS achieves an earlier time-to-accuracy, and the gap widens under the more clustered GM configuration \((14,2.5)\). 
Spatial clustering skews the distance distribution and inflates the maximum uplink time for $\mathrm{CONV}$, whereas PASS approximately sets \(r_m=d\), contracting both the mean and the dispersion of per-round latencies. In AFL (Fig.~\ref{fig:spatial-afl}), PASS curves for all spatial distributions nearly coincide and converge rapidly, while $\mathrm{CONV}$ is noticeably slower and more distribution dependent, consistent with PASS suppressing event-latency dispersion so that wall-clock progress is governed by trigger/compute rather than communication heterogeneity.

\subsection{Discussion}
Given the analysis and simulations above, pinching the active radiator shortens the worst link and suppresses the latency tail, translating into faster wall‑clock learning. 
The considered uplink model activates one PA per uplink (AFL) and a single midpoint PA per SFL round, avoiding multi‑PA re‑radiation on a continuous waveguide and yielding a tractable rate/latency model.

In the downlink PASS, it is well captured by feed‑to‑waveguide EM coupling with passive radiation from activated PAs \cite{Ding_PASS_Original, Ouyang_array_gain}. However, uplink is non-trivial to be modelled because each PA can receive free‑space energy into the waveguide and re‑radiate via other PAs along a continuous line, complicating physically consistent multi‑PA uplink models \cite{ouyang_segment}.
A practical implementation used in our analysis is one PA per waveguide, which removes in‑waveguide re‑radiation, keeps the model tractable, and already delivers the measured latency/participation/convergence gains.
To fully exploit PASS without sacrificing uplink tractability and introducing in-waveguide attenuation, segmented waveguide‑enabled architectures (SWAN‑style designs) split a long waveguide into separately fed segments. Such design activates at most one PA per segment eliminates inter‑segment re‑radiation while shortening PA–user and PA–feed distances, which can serve as a promising direction for the uplink wireless FL \cite{ouyang_segment}.

\section{Conclusion}
We have presented a comprehensive study of FL in wireless networks from a PA perspective, developing both theoretical insights and system-level evaluations. Our analysis confirms that PASS can fundamentally alleviate the wireless straggler bottleneck in FL. In SFL, we proved that pinching the antenna to users strictly reduces the worst-case link distance, leading to uniformly faster round times. In AFL, we showed that PASS increases timely user participation and effectively suppresses long tail delays, shifting the training regime from communication-limited to computation-limited. We also integrated these physical layer gains into an FL convergence analysis, finding that PASS raises the floor for user inclusion probabilities and lowers variance and quantization error floors, yielding faster convergence in wall-clock time.
The simulation results validated our theoretical findings: PASS consistently accelerated model training and improved reliability compared to $\mathrm{CONV}$. In summary, by tackling stragglers at the physical layer, PASS offers a new degree of freedom to speed up FL. This work provides a foundational understanding of that benefit. Future research can investigate combining PASS with advanced scheduling or adaptive algorithms to further enhance performance in practical deployments.

\appendices
\section{The Mean and Second Moment of \texorpdfstring{$Y_{[M]}$}{Y_{[M]}}} \label{appendix:Ym}
By normalizing via $\widetilde Y_i:=\frac{2}{D}Y_i\sim\mathcal{U}[0,1]$, we obtain $\widetilde Y_{[M]}\sim \mathrm{Beta}(M,K{+}1{-}M)$ with probability density function:
\begin{equation}\label{eq:beta-pdf}
f_{\widetilde Y_{[M]}}(u)=\frac{K!}{(M{-}1)!(K{-}M)!}\,u^{M-1}(1-u)^{K-M}, u\in[0,1],
\end{equation}
and moments given by \cite[Ch.~2]{DavidNagaraja03}):
\begin{equation}
\mathbb E[\widetilde Y_{[M]}]=\frac{M}{K+1},\quad 
\mathbb E[\widetilde Y_{[M]}^2]=\frac{M(M+1)}{(K+1)(K+2)}.
\label{eq:beta-moments}
\end{equation}

\section{Proof of Concentration for \texorpdfstring{$M$}{M}-th closest user}\label{proof_lemma_1}
\begin{remark}[Tail concentration of the $M$-th closest user]
The normalized bottleneck offset $\widetilde Y_{[M]}$ concentrates exponentially around its mean $M/(K{+}1)$.
\end{remark}

\begin{lemma}[Concentration for the $M$-th closest user]
\label{lem:conv-concentration-corrected}
Let $\widetilde Y_{[M]}$ be the $M$-th order statistic of $K$ i.i.d.\ $\mathcal U[0,1]$ and set
$p^\dagger := \frac{M}{K+1}$, $p^\star := \frac{M}{K}$. For any $\varepsilon \in (0,\min\{p^\dagger,1-p^\dagger\})$,
\begin{align}
\mathbb{P}\big(\widetilde Y_{[M]} \le p^\dagger - \varepsilon\big)
&\le \exp\!\Big(-K\,D\big(p^\star \,\big\|\, p^\dagger-\varepsilon\big)\Big), \tag{B1a}\\
\mathbb{P}\big(\widetilde Y_{[M]} \ge p^\dagger + \varepsilon\big)
&\le \exp\!\Big(-K\,D\big(\tfrac{M-1}{K} \,\big\|\, p^\dagger+\varepsilon\big)\Big), \tag{B2b}
\end{align}
hence
\begin{equation}
\begin{aligned}
\mathbb{P}\big(|\widetilde Y_{[M]}-p^\dagger|\ge \varepsilon\big)
&\le \exp\!\Big(-K\,D\big(p^\star \,\big\|\, p^\dagger-\varepsilon\big)\Big)
\\&\quad
+ \exp\!\Big(-K\,D\big(\tfrac{M-1}{K} \,\big\|\, p^\dagger+\varepsilon\big)\Big).
\end{aligned}\tag{B3c}
\end{equation}
Moreover, a simple symmetric bound centered at the mean is
\begin{equation}
\mathbb{P}\big(|\widetilde Y_{[M]}-p^\dagger|\ge \varepsilon\big) \le 2\exp(-2K\varepsilon^2), \qquad (0<\varepsilon<1). \tag{B4d}
\end{equation}
All KL divergences are in nats. 
\end{lemma}

\begin{IEEEproof}
We interpret order statistics via a Binomial counting argument. The distributional identities
\begin{equation}
\begin{aligned}
& \mathbb{P}(\widetilde Y_{[M]}\le u)=\mathbb{P}\!\big(\mathrm{Bin}(K,u)\ge M\big),\quad \\
&
\mathbb{P}(\widetilde Y_{[M]}\ge u)=\mathbb{P}\!\big(\mathrm{Bin}(K,u)\le M{-}1\big)
\end{aligned}
\end{equation}
are standard for uniform order statistics. Applying Chernoff bounds for a Binomial random variable gives (5a)–(5b):
$\mathbb{P}(\mathrm{Bin}(K,u)\ge Ka)\le e^{-K D(a\|u)}$ for $a>u$ and
$\mathbb{P}(\mathrm{Bin}(K,u)\le Ka)\le e^{-K D(a\|u)}$ for $a<u$, with
$u=p^\dagger\mp\varepsilon$ and $a\in\{p^\star,(M{-}1)/K\}$.
For \textup{(5d)}, write $\hat p:=\mathrm{Bin}(K,u)/K$. Then
\begin{equation}
\begin{aligned}
&\mathbb{P}(\widetilde Y_{[M]} \le p^\dagger-\varepsilon)
=\mathbb{P}\big(\hat p - u \ge p^\star - (p^\dagger-\varepsilon)\big) \\
& \qquad \qquad \qquad \qquad \le \exp\!\big(-2K\,[p^\star - p^\dagger + \varepsilon]^2\big),
\end{aligned}
\end{equation}
by Hoeffding’s inequality. Since $p^\star - p^\dagger = \frac{M}{K}-\frac{M}{K+1}=\frac{M}{K(K+1)}\ge 0$, the exponent is at most $-2K\varepsilon^2$. Similarly,
\begin{equation}
\begin{aligned}
& \mathbb{P}(\widetilde Y_{[M]} \ge p^\dagger+\varepsilon)
=\mathbb{P}\big(u-\hat p \ge (p^\dagger+\varepsilon) - \tfrac{M-1}{K}\big) \\
& \qquad \qquad \qquad \qquad \le \exp\!\big(-2K\,[p^\dagger - \tfrac{M-1}{K} + \varepsilon]^2\big),
\end{aligned}
\end{equation}
and $p^\dagger - \frac{M-1}{K}=\frac{K+1-M}{K(K+1)}\ge 0$, giving the same upper bound $\exp(-2K\varepsilon^2)$. 
\end{IEEEproof}

\section{Proof of Proposition~\ref{prop:PA-LB}}
\label{app:PA_LB_proof}

Let $U_{(1)}\le\cdots\le U_{(K)}$ be the order statistics of $K$ i.i.d.\ $\mathcal U[0,1]$ and define the simple spacings
$G_1:=U_{(1)}$, $G_j:=U_{(j)}-U_{(j-1)}$ for $2\le j\le K$, $G_{K+1}:=1-U_{(K)}$,
so that $(G_1,\dots,G_{K+1})\sim\mathrm{Dirichlet}(1,\dots,1)$.
Write $M_*:=\min_{1\le j\le K+1}G_j$.
For the uniform Dirichlet (i.e., uniform on the simplex), the survival function of $M_*$ admits a closed form:
\[
\mathbb P(M_*\ge t)=\left(1-(K{+}1)t\right)^K,\quad 0\le t\le \tfrac{1}{K{+}1},
\]
and $0$ otherwise. Differentiating,
\[
f_{M_*}(t)=K(K{+}1)\left(1-(K{+}1)t\right)^{K-1},\quad 0\le t\le \tfrac{1}{K{+}1}.
\]
Hence
\[
\mathbb E[M_*^2]
=\int_{0}^{\frac{1}{K+1}} t^2 f_{M_*}(t)\,dt
=\frac{K}{(K+1)^2}\int_{0}^{1} (1-u)^2 u^{K-1}\,du
\]
which further equals $=\frac{K}{(K+1)^2}\,\mathrm{Beta}(K,3)$,
where $u:=1-(K{+}1)t$ and $\mathrm{B}$ is the Beta function.
Using $\mathrm{B}(K,3)=\frac{2}{K(K+1)(K+2)}$, we obtain
\[
\mathbb E[M_*^2]=\frac{2}{(K+1)^3(K+2)}.
\]
Since every $m$-span is a sum of $m$ positive spacings, $\widetilde L_M\ge m M_*$ and therefore
\[
\mathbb E\big[\widetilde L_M^2\big]\ \ge\ m^2\,\mathbb E[M_*^2]
= m^2\,\frac{2}{(K+1)^3(K+2)}.
\]
Finally, $x_{\rm str}^{\rm PA}= \frac{D}{2}\,\widetilde L_M$ implies
\(
\mathbb E\!\big[x_{\rm str}^{\rm PA\,2}\big]
\ \ge\ \frac{D^2}{4}\,m^2\,\frac{2}{(K+1)^3(K+2)},
\)
which yields \eqref{eq:PA-LB}.

\section{Details for Eq. (\ref{eq:PA-scaling-bounds})}
\label{app:PA_consequences}
From \eqref{eq:PA-UB-beta} and \eqref{eq:conv-YM-second-moment},
\[
\frac{\mathbb E[x_{\rm str}^{\rm PA\,2}]}{\mathbb E[Y_{[M]}^2]}
\le \frac{m(m+1)}{M(M+1)}=\frac{M-1}{M+1}.
\]
Using \eqref{eq:PA-UB-avg} instead yields, as $K\to\infty$ with fixed $M$,
\[
\limsup_{K\to\infty}\frac{\mathbb E[x_{\rm str}^{\rm PA\,2}]}{\mathbb E[Y_{[M]}^2]}
\le \frac{m^2}{M(M+1)}.
\]
The scaling \eqref{eq:conv-scaling} is immediate from \eqref{eq:conv-YM-second-moment}.
The bounds \eqref{eq:PA-scaling-bounds} follow from the PA UB \eqref{eq:PA-UB-beta} and LB \eqref{eq:PA-LB}.

\section{High-SNR preliminaries and uniform threshold}
\label{app:hsnr}

For $x\in[-\tfrac{D}{2},\tfrac{D}{2}]$, we have $x^2{+}d^2\in\big[d^2,\,d^2{+}(\tfrac{D}{2})^2\big]$, hence
\[
C_0:=\max\!\Big\{\,\big|\log_2 d^2\big|,\ \big|\log_2\!\big(d^2{+}(\tfrac{D}{2})^2\big)\big|\,\Big\},
\]
and $C_1:=\frac{d^2+(\tfrac{D}{2})^2}{\ln 2}$, then using $\ln(1+u)\le u$ for $u\ge 0$ we obtain uniformly in $x$,
\[
|\delta(x)|\ \le\ C_0+\frac{\ln(1+v(x))}{\ln 2}\ \le\ C_0+C_1\,2^{-\Lambda}.
\tag{A.1}
\label{eq:app-delta-bound}
\]

Then, for an explicit threshold for geometric expansion, choose
\[
\Lambda_0 := \max\Big\{\,4C_0,\ \log_2(4C_1),\ 1\Big\}.
\tag{A.2}
\label{eq:app-Lambda0}
\]
Then for all $\Lambda\ge\Lambda_0$ and all $x\in[-\tfrac{D}{2},\tfrac{D}{2}]$,
\(
\frac{|\delta(x)|}{\Lambda}\ \le\ \frac{1}{2}
\).

\begin{lemma}[Uniform expansion with explicit remainder]
\label{lem:uniform-bridge}
For all $\Lambda\ge \Lambda_0$,
\begin{equation}
\label{eq:uniform-bridge}
\frac{1}{R_0(x)}\ =\ \frac{1}{\Lambda}+\frac{\log_2(x^2{+}d^2)}{\Lambda^2}\ +\ \mathcal R_\Lambda(x),
\end{equation}
where the remainder is uniformly bounded on $[-\frac{D}{2},\frac{D}{2}]$ by:
\begin{equation}
\label{eq:uniform-remainder}
\big|\mathcal R_\Lambda(x)\big|
\ \le\ \frac{2\,\big(C_0+C_1 2^{-\Lambda}\big)^2}{\Lambda^3}
\ +\ \frac{C_1\,2^{-\Lambda}}{\Lambda^2}\,.
\end{equation}
Consequently, $\sup_x |\mathcal R_\Lambda(x)|=O(\Lambda^{-3})$ and expectations may be
interchanged with \eqref{eq:uniform-bridge} by dominated convergence.
\end{lemma}
\begin{IEEEproof}
    Refer to Appendix \ref{proof_l2}.
\end{IEEEproof}
It ensures the geometric-series expansion
\(
\frac{1}{\Lambda+\delta}=\frac{1}{\Lambda}\sum_{k=0}^\infty(-\delta/\Lambda)^k
\)
used in Lemma~\ref{lem:uniform-bridge}. In particular, this yields a uniform
$O(\Lambda^{-3})$ remainder and allows interchanging expectation with the expansion in the main text.

\section{Proof of Lemma 2}\label{proof_l2}
\begin{IEEEproof}
For $|\delta|/\Lambda\le 1/2$,
\[
\frac{1}{\Lambda+\delta}
=\frac{1}{\Lambda}\sum_{k\ge 0}(-\delta/\Lambda)^k
=\frac{1}{\Lambda}-\frac{\delta}{\Lambda^2}+\frac{1}{\Lambda}\sum_{k\ge 2}(-\delta/\Lambda)^k.
\]
Since $\sum_{k\ge 2} r^k=r^2/(1-r)\le 2r^2$ for $0\le r\le 1/2$,
\[
\Big|\frac{1}{\Lambda}\sum_{k\ge 2}(-\delta/\Lambda)^k\Big|
\le \frac{2\,\delta^2}{\Lambda^3}.
\]
Furthermore,
\(
-\frac{\delta}{\Lambda^2}
=\frac{\log_2(x^2{+}d^2)}{\Lambda^2}-\frac{\ln(1+v(x))}{\Lambda^2\ln 2},
\)
and $\ln(1+v)\le v\le (d^2+(\frac{D}{2})^2)2^{-\Lambda}$ gives the second term in
\eqref{eq:uniform-remainder}. Using $|\delta|\le C_0+C_1 2^{-\Lambda}$ yields the first term.
\end{IEEEproof}

~\bibliographystyle{IEEEtran}
\bibliography{PASS_FL_Ref}
\end{document}